\documentclass{amsart}
\usepackage{amsmath}
\usepackage{amssymb}
\usepackage{amsthm}
\usepackage{amsfonts}
\usepackage{amstext}
\usepackage{amsopn}
\usepackage{amsxtra}
\usepackage{mathrsfs}
\usepackage{graphicx}
\usepackage[latin1]{inputenc}
\usepackage{dsfont}
\usepackage{pst-all}
\newcommand\R{{\ensuremath {\mathbb R} }}
\newcommand\C{{\ensuremath {\mathbb C} }}
\newcommand\N{{\ensuremath {\mathbb N} }}
\newcommand\Z{{\ensuremath {\mathbb Z} }}

\newcommand\bE{{\ensuremath {\mathbb E}}}
\newcommand\1{{\ensuremath {\mathds 1} }}
\newcommand\bP{{\ensuremath {\mathbb P} }}
\renewcommand\phi{\varphi}

\newcommand{\bH}{\mathbb{H}}

\newcommand{\ccL}{\mathscr{L}}

\renewcommand{\to}{\rightarrow}

\newcommand{\cN}{\mathcal{N}}

\newcommand{\cJ}{\mathcal{J}}

\newcommand{\cF}{\mathcal F}

\newcommand{\cR}{\mathcal R}

\newcommand{\cK}{\mathcal{K}}

\newcommand{\cG}{\mathcal{G}}

\newcommand\ii{{\ensuremath {\infty}}}
\newcommand\pscal[1]{{\ensuremath{\left\langle #1 \right\rangle}}}
\newcommand{\norm}[1]{ \left| \! \left| #1 \right| \! \right| }

\DeclareMathOperator{\tr}{{\rm Tr}}

\renewcommand{\leq}{\leqslant}
\renewcommand{\geq}{\geqslant}

\renewcommand{\epsilon}{\varepsilon}

\newtheorem{theorem}{Theorem}
\newtheorem{lemma}[theorem]{Lemma}
\newtheorem{proposition}[theorem]{Proposition}
\newtheorem{remark}[theorem]{Remark}
\newtheorem{definition}[theorem]{Definition}
\newtheorem{corollary}[theorem]{Corollary}
\newtheorem{example}[theorem]{Example}

\begin{document}
 
\author[X. Blanc]{Xavier BLANC}
\address{CEA, DAM, DIF, F-91297 Arpajon, France.}
\email{blanc@ann.jussieu.fr, xavier.blanc@cea.fr}

\author[M. Lewin]{Mathieu LEWIN}
\address{CNRS \& Laboratoire de Mathématiques (UMR 8088), Université de Cergy-Pontoise, F-95000 Cergy-Pontoise, France.}
\email{mathieu.lewin@math.cnrs.fr}

\title[Thermodynamic limit for disordered quantum Coulomb systems]{Existence of the thermodynamic limit for disordered quantum Coulomb systems}

\maketitle

\begin{center}
\it Dedicated to Elliott H. Lieb, on the occasion of his 80th birthday
\end{center}

\medskip

\begin{abstract}
Following a recent method introduced by C. Hainzl, J.P. Solovej and the second author of this article, we prove the existence of the thermodynamic limit for a system made of quantum electrons, and classical nuclei whose positions and charges are randomly perturbed in an ergodic fashion. All the particles interact through Coulomb forces.

\medskip 

\noindent{\scriptsize\copyright~2012 by the authors. This paper may be reproduced, in its entirety, for non-commercial purposes. Final version to appear in \emph{J. Math. Phys.}}
\end{abstract}

\bigskip

\section{Introduction}

One of the main purposes of Statistical Physics is to understand the macroscopic behavior of microscopic systems. For regular matter, composed of negative (electrons) and positive (nuclei) charges, this question is highly non trivial because of the long range of the Coulomb potential.

In 1966, Fisher and Ruelle have in~\cite{FisRue-66} raised the important question of the \emph{stability} of many-particle systems at the \emph{macroscopic scale}. This may be formulated by requiring that the energy per particle (or the energy per unit volume) stays bounded from below when the number of particles (or the volume $|D|$ of the sample) is increased,
$$\frac{\cF(D)}{|D|}\geq -C.$$ 
For many-body systems interacting through Coulomb forces like in ordinary matter, the first proof of stability is due to Dyson and Lenard~\cite{DysLen-67,DysLen-68}. The fermionic nature of the electrons is then important~\cite{Dyson-67}. A different proof was later found by Lieb and Thirring in~\cite{LieThi-75}, based on a celebrated inequality which now carries their name. We refer the reader to~\cite{Lieb-76,Lieb-90,LieSei-09} for a review of results concerning the stability of matter.

The stability of matter as defined by Fisher and Ruelle only shows that
the system does not collapse when the number of particles grows. A more precise 
requirement is that the energy per particle (or the
energy per unit volume) actually \emph{has a limit} when the number of
particle (or the volume $|D|$) goes to infinity
$$\lim_{|D|\to\ii}\frac{\cF(D)}{|D|}=f.$$
For short-range interactions,
this was already done by Ruelle~\cite{Ruelle-63a,Ruelle-63b,Ruelle} and
Fisher~\cite{Fisher-64}. The first proof of a theorem of this form for
Coulomb systems is due to Lieb and Lebowitz in~\cite{LieLeb-72}. In this latter work rotational invariance plays a crucial role. For quantum crystals, in which the nuclei are classical particles clamped on a lattice (a system which is obviously not rotationally invariant), the first proof goes back to Fefferman~\cite{Fefferman-85}. The main challenge of all these works was to find an adequate way to prove the existence of \emph{screening}, the fact that matter spontaneously organizes in a locally neutral way. Screening is at the origin of a faster decay of the interactions between the particles and it is the main explanation for the existence of such systems at the macroscopic scale. The importance of screening was already stressed in a fundamental paper of Onsager~\cite{Onsager-39}.

In two recent papers~\cite{HaiLewSol_1-09,HaiLewSol_2-09}, Hainzl, Solovej and the second author of this article have proposed a new method for proving the existence of the thermodynamic limit for quantum Coulomb system. This method is based on an inequality quantifying screening due to Graf and Schenker~\cite{GraSch-95}, and which was itself inspired of earlier works by Conlon, Lieb and Yau~\cite{ConLieYau-88,ConLieYau-89}. The purpose of the present work is to extend the results of~\cite{HaiLewSol_1-09,HaiLewSol_2-09} to \emph{stochastic systems} in which the electrons are quantum and the nuclei are \emph{random classical particles}. 

It has been known for a long time that the presence of disorder can strongly influence the behavior of a quantum system. The most famous example is of course the so-called {Anderson localization}~\cite{Anderson-58} of particles under weak disorder. On the mathematical side, lots of works have been devoted to the study of \emph{noninteracting} disordered quantum systems, for instance described by random Schrödinger operators (for an introduction to these results, see, e.g.,~\cite{Hislop-08,Kirsch-08}). To our knowledge, the mathematical literature on \emph{interacting} disordered many-body systems is quite limited, in spite of the increasing physical interest devoted to such systems~\cite{BasAleAlt-06,Deissler-10,Sanchez-10}. In recent works~\cite{ChuSuh-10,AizWar-09} localization bounds were derived for systems of a finite number of particles with short range interactions. Some authors considered nonlinear random models describing condensed bosonic systems, mainly in Gross-Pitaevskii theory (see, e.g.,~\cite{KloMet-11} and the references therein). 

In a recent paper~\cite{Veniaminov-11}, Veniaminov has initiated the mathematical study of the thermodynamic limit of random many-body quantum systems with short range interactions, following the approach of Ruelle and Fisher. His work does not cover ordinary matter made of Coulomb charges, however. Large stochastic Coulomb systems were considered before by Le Bris, Lions and the first author of this paper in~\cite{BlaBriLio-07}. There the electrons are only described by Thomas-Fermi-type theories and, in this case, it is possible to identify the thermodynamic limit $f$ exactly. 

Following the method of~\cite{HaiLewSol_1-09,HaiLewSol_2-09}, we are able to deal with quantum electrons satisfying the full many-body Schrödinger equation, in the Coulomb field of a random distribution of pointwise classical nuclei. Similarly as in~\cite{BlaBriLio-07}, we typically think of a perfect (periodic) lattice of nuclei whose location and charges are perturbed randomly. For the existence of the limit, we have to assume that this randomness has some translation invariance. This is reflected in the assumption that the distribution of nuclei is stationary and ergodic, as we explain below. 

In the next section we properly define our model and we state our main theorem. In short, it says that the thermodynamic limit
exists and is deterministic, that is, independent of the randomness $\omega$:
\begin{equation}
\lim_{|D|\to\ii}\frac{\cF(\omega,D)}{|D|}=f.
\label{eq:thermo_intro}
\end{equation}
More precisely, the limit~\eqref{eq:thermo_intro} holds in some $L^p$ space with respect to the randomness
\begin{equation}
\lim_{|D|\to\ii}\bE\left|\frac{\cF(\cdot,D)}{|D|}-f\right|^p=0,\quad p\geq1.
\label{eq:thermo_intro_L_p}
\end{equation}
Almost-sure convergence is expected as well, but not proved in this paper. In general, the convergence cannot be uniform with respect to $\omega$. In Section~\ref{sec:nuclei} we will show that if we attach independent harmonic oscillators to the nuclei of a cubic lattice, and make them vibrate randomly according to the associated Gibbs measure, then 
$$\bE\big(\cF(\cdot,D)^3\big)=+\ii$$
for any $D$ large enough. Therefore, \eqref{eq:thermo_intro_L_p} cannot hold
for $p\geq 3$ in general.

For the proof of~\eqref{eq:thermo_intro_L_p}, we will rely heavily on the machinery introduced in~\cite{HaiLewSol_1-09,HaiLewSol_2-09} for deterministic systems. Some parts of the proof which are similar to those of \cite{HaiLewSol_1-09,HaiLewSol_2-09} will only be sketched. 

\medskip

\subsection*{Acknowledgement}
The authors acknowledge support from the French Ministry of Research (ANR-10-BLAN-0101). M.L. acknowledges support from the  European Research Council under the European Community's Seventh Framework Programme (FP7/2007-2013 Grant Agreement MNIQS 258023).

\section{Main result}

\subsection{Random distribution of the nuclei}
We consider a fixed discrete subgroup $\ccL$ of $\R^3$, with bounded fundamental domain $W$. The whole space $\R^3$ is the disjoint union of the sets $W+j$ for $j\in\ccL$. We typically think of $\ccL=\Z^3$ with $W=[-1/2,1/2)^3$ the semi-open unit cube. 

Our nuclei are placed in random locations in the whole space $\R^3$. We assume that the probability distribution of their positions and charges in a given domain $D$ is the same when $D$ is translated by a vector $k\in\ccL$. The appropriate mathematical notion is that of \emph{stationarity} which we now recall.

Let $(\Omega,{\mathscr T},\bP)$ be a probability space. We assume that the discrete group $\ccL$ acts on $\Omega$ and we denote this action by $\tau_k$ for $k\in\ccL$. In the whole paper, the group action is supposed to be measure preserving, 
\begin{equation}
\forall k\in\ccL,\ \forall T\in{\mathscr T},\quad \bP(\tau_k T)=\bP(T),
\label{eq:measure-preserving} 
\end{equation}
and ergodic
\begin{equation}
\big(\tau_k T=T,\ \forall k\in\ccL\big)\ \Longrightarrow\ \bP(T)\in\{0,1\}.
\label{eq:ergodic} 
\end{equation}

We now follow the notation of~\cite{HaiLewSol_1-09} and describe our nuclei by a countable set $\cK=\{(R,z)\}\subset\R^3\times[\underline{Z},\bar{Z}]$ with $0<\underline{Z}<\bar{Z}$. We always make the assumption that the nuclei have a highest possible charge $\bar{Z}$. Also, their charge cannot be smaller than $\underline{Z}$. In reality the charges of the nuclei are integers and they are smaller than $118$.
In our random setting, the set $\cK$ is random, that is, it depends on $\omega\in\Omega$:
$$\cK(\omega)=\big\{(R_j(\omega),z_j(\omega)),\ R_j(\omega)\in\R^3,\ z_j(\omega)\in[\underline{Z},\bar{Z}],\ j\in\N\big\}.$$
The specific choice of a numbering of these nuclei by the index $j$ has no real importance. 
Our main assumption is that the sets $\cK(\omega)$ are \emph{stationary} with respect to the action of $\ccL$ on $\Omega$,  in the sense that
\begin{equation}
\forall k\in\ccL,\qquad\cK(\tau_k\omega)=\cK(\omega)-k:=\big\{(R-k,z)\ :\ (R,z)\in\cK(\omega)\big\}.
\label{eq:stat_K}
\end{equation}
A similar setting is used in~\cite{BlaBriLio-07}.
Throughout the paper we will always make the hypothesis that the number of nuclei in any given set $D$ is bounded almost surely, and that two nuclei can never be on top of each other. By stationarity, this means that the random variable
\begin{equation}
\boxed{X_0(\omega):=\#\big\{(R,z)\in\cK(\omega)\ :\ R\in W\big\}}
\label{eq:def_X_0}
\end{equation}
is almost surely finite. 
It will be convenient to introduce the distance to the nearest neighbor of each nuclei in $\cK(\omega)$, which is the random variable
\begin{equation}
\forall (R,z)\in\cK(\omega),\qquad \delta_{R,z}(\omega):=\inf_{\substack{(R',z')\in\cK(\omega)\\ R'\neq R}}|R-R'|.
\end{equation}
Our assumption that all the nuclei are different means that $\inf\{\delta_{R,z}(\omega)\ :\ (R,z)\in \cK(\omega)\cap{W}\}>0$ almost surely. Since the number of nuclei is locally finite almost surely, this is the same as asking that the random variable
\begin{equation}
\boxed{X_1(\omega):=\sum_{\substack{(R,z)\in\cK(\omega)\\ R\in{W}}}\frac{1}{\delta_{R,z}(\omega)}}
\label{eq:def_X_1}
\end{equation}
is almost surely finite.
In the following we will write for simplicity $(R,z)\in \cK(\omega)\cap D$ to say that $(R,z)\in\cK(\omega)$ and $R\in D$.

\begin{example}[The i.i.d. case]\label{ex:iid} 
The simplest example to keep in mind is that of nuclei on a lattice which are perturbed by independent and identically distributed (i.i.d.) random variables (see Figure~\ref{fig:dessin_noyaux}). We explain this for a cubic cristal $\ccL=\Z^3$ with exactly one nucleus per unit cell. The extension to the general setting is straightforward.

Let us fix a probability space $(\Omega_0,\mathscr{T}_0,\bP_0)$ and consider the product space
$$\Omega=\big(\Omega_0\big)^{\Z^3},\quad \mathscr{T}=\sigma(\mathscr{T}_0)^{\Z^3}, \quad\bP=(\bP_0)^{\otimes\Z^3}.$$ 
We choose for the action of $\ccL$ on $\Omega$ the shift $\tau_k\big[(\omega_j)_{j\in\Z^3}\big]=(\omega_{j+k})_{j\in\Z^3}$. It is known to be ergodic. Consider then $r:\Omega_0\to\R^3$ and $z:\Omega_0\to[\underline{Z},\bar{Z}]$ two fixed random variables. The families of i.i.d. random variables 
$$r_j(\omega):=r(\omega_j),\qquad z_j(\omega):=z(\omega_j)$$
are stationary in the sense that $r_j(\tau_k\omega)=r(\omega_{j+k})=r_{j+k}(\omega)$ (and a similar property for $z_j$). Finally, we let
\begin{equation}
\cK(\omega)=\Big\{ \big(j+ r_j(\omega)\,,\,z_j(\omega)\big)\ :\ j\in\Z^3\Big\}.
\label{eq:iid}
\end{equation}
It is obvious that $\cK$ is stationary in the sense of~\eqref{eq:stat_K}. 
We typically think of a Gaussian random variable $r$ whose law is given by 
\begin{equation}
\nu(x)=\frac{1}{(2\pi \sigma)^{3/2}}e^{-\frac{|x|^2}{2\sigma}},
\label{eq:Gaussian}
\end{equation}
and which corresponds to independent harmonic vibrations of the nuclei.

The number of nuclei in $W$ is given by
$$X_0(\omega)=\sum_{j\in\Z^3}\1\big(j+r_j(\omega)\in W\big).$$
We have by stationarity
$$\bP(j+r_j\in W)=\bP(r_0\in W-j)=\bP_0(r\in W-j)$$
and therefore
$$\bE\big(X_0\big)=\sum_{j\in\Z^3}\bP(j+r_j\in W)=\sum_{j\in\Z^3}\bP_0(r\in W-j)=1$$
Hence $X_0\in L^1(\Omega)$ is finite almost-surely. The average number of nuclei per unit cell is the same as in the deterministic case, it is independent of $\nu$. It is possible to give conditions on the random displacement $r$ which ensure that $X_1$ is finite as well (see Section~\ref{sec:nuclei}). 

In this example we have assumed that no nucleus is ever removed from the system. The opposite case can be handled by allowing $z(\omega)\in\{0\}\cup[\underline{Z},\bar{Z}]$ and adding the assumption that $z_j(\omega)\neq0$ in the definition of $\cK(\omega)$. Then $\bE(X_0)=\bP_0(z\neq0)$.
\end{example}

\begin{figure}[h]
\includegraphics{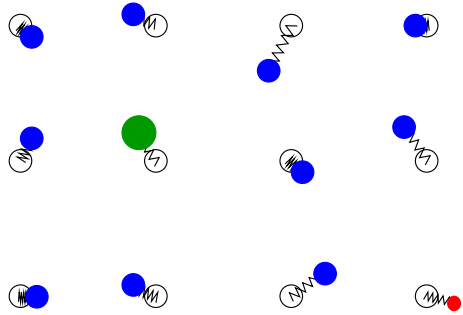}

\medskip

\caption{The case of nuclei on a lattice whose charges and positions are perturbed by i.i.d. random variables. A Gaussian displacement~\eqref{eq:Gaussian} corresponds to having the nuclei attached to harmonic oscillators vibrating randomly and independently, according to their Gibbs measure. \label{fig:dessin_noyaux}}
\end{figure}

As we will see later, in the general case the conditions that $X_0$ and $X_1$ are finite almost-surely are not at all enough to prove the existence of the thermodynamic limit. But, before writing a more precise condition, we turn to the description of the quantum electrons.

\subsection{The grand canonical free energy of the electrons}
Here we mainly follow~\cite{HaiLewSol_2-09}. Let $D$ be a bounded open subset of $\R^3$. The Hamiltonian for $N$ electrons in $D$ is the random self-adjoint operator
\begin{multline}
H(\omega,D,N)=\sum_{n=1}^N(-\Delta)_{x_n}-\sum_{n=1}^N\;\sum_{(R,z)\in\cK(\omega)\cap D}\frac{z}{|x_n-R|}\\[0,2cm]
+\sum_{1\leq n<m\leq N}\frac{1}{|x_n-x_m|}+\frac12\sum_{\substack{(R,z)\in\cK(\omega)\cap D\\ (R',z')\in\cK(\omega)\cap D\\ R\neq R'}}\frac{z\,z'}{|R-R'|}.
\label{eq:Hamiltonian}
\end{multline}
Here $-\Delta$ is the Laplacian with Dirichlet boundary conditions in $D$. The operator $H(\omega,D,N)$ acts on the fermionic space $\bigwedge_1^N L^2(D)$ consisting of square-integrable functions $\Psi(x_1,...,x_N)$ which are antisymmetric with respect to exchanges of the $x_j$. The form domain of $H(\omega,D,N)$ is the Sobolev space $\bigwedge_1^N H^1_0(D)$. Since there is a finite number of nuclei in $D$, which are all distinct, for almost every $\omega\in\Omega$, the operator $H(\omega,D,N)$ is self-adjoint on $\bigwedge_1^NH^2(D)\cap H^1_0(D)$ and bounded from below. For simplicity we have chosen units in which $\hbar=1$ and the mass $m$ and charge of the electrons are $m=1/2$, $e=1$. We have also neglected the spin for convenience. The terms in $H(\omega,D,N)$ respectively account for the kinetic energy of the electrons, the nuclei/electrons attraction, the electrons/electrons repulsion, and the nuclei/nuclei repulsion. 

At zero temperature, the ground state energy for $N$ electrons in $D$ is the random variable
\begin{equation}
{\cF_0(\omega,D,N)=\inf\sigma\big(H(\omega,D,N)\big)=\inf_{\substack{\Psi\in \bigwedge_1^N H^1_0(D)\\ \int_{D^N}|\Psi|^2=1}}\pscal{\Psi\,,\,H(\omega,D,N)\Psi}.}
\label{eq:def_energy_N_elec}
\end{equation}
Since $H(\omega,D,N)$ is bounded from below, $\cF_0(\omega,D,N)$ is a well-defined random variable. As we will explain, the stability of matter tells us that there is a lower bound on $\cF_0(\omega,D,N)$ which is independent of $\omega$ and $N$ (see Theorem~\ref{thm:stability} below).
The stationarity of the distribution of the nuclei implies a certain stationarity property of $\cF_0$ which reads
$\cF_0(\tau_k \omega,D,N)=\cF_0(\omega,D-k,N).$
Indeed, the Hamiltonians $H(\tau_k\omega,D,N)$ and $H(\omega,D-k,N)$ are isometric and therefore they have the same spectrum.

Like in~\cite{Fefferman-85,HaiLewSol_2-09}, we work in the grand canonical ensemble, that is, we optimize over the number of electrons, instead of imposing the neutrality of the system. The ground state energy in $D$ for a realization of the distribution of the nuclei is then defined as
\begin{equation}
\boxed{\phantom{\int}\cF_0(\omega,D):=\inf_{N\geq0}\cF_0(\omega,D,N).\phantom{\int}}
\label{eq:def_GS_energy}
\end{equation}
This random variable is stationary in the sense that 
\begin{equation}
\cF_0(\tau_k \omega,D)=\cF_0(\omega,D-k) 
\label{eq:stat-E-N}
\end{equation}
for any chosen domain $D$, all $k\in\ccL$ and almost all $\omega\in\Omega$.

At positive temperature $T>0$ with chemical potential $\mu\in\R$, the grand canonical free energy is defined by the formula
\begin{equation}
\boxed{\cF_{T,\mu}(\omega,D)=-T\log\left(\sum_{N\geq0}\tr_{\bigwedge_1^NL^2(\R^3)}e^{-\big(H(\omega,D,N)-\mu\,N\big)/T}\right)}
\label{eq:def_free_energy_N_elec}
\end{equation}
and it satisfies a similar stationarity property as $\cF_0$.

There is a useful variational formula for $\cF_{T,\mu}$ in the fermionic Fock space $\mathscr{F}:=\C\oplus\bigoplus_{N\geq1}\bigwedge_1^N L^2(D)$. Introducing the operators 
$$\bH(\omega,D):=0\oplus\bigoplus_{N\geq1}H(\omega,D,N)\quad\text{and}\quad \cN:=0\oplus\bigoplus_{N\geq1}N,$$
we then have
\begin{equation}
\cF_{T,\mu}(\omega,D)=\inf_{\substack{\Gamma\geq0\\ \tr_{\mathscr{F}}\Gamma=1}} \tr_{\mathscr{F}}\Big(\big(\bH(\omega,D)-\mu\cN\big)\Gamma+T\;\Gamma\log\Gamma\Big).
\label{eq:variational} 
\end{equation}
Here $\Gamma$ is the density matrix of a mixed quantum state in Fock space. See~\cite{HaiLewSol_2-09} for more details.

\subsection{Existence of the thermodynamic limit}
Before stating our main result, we quote the following important lower bound on $\cF_{T,\mu}$, which is nothing else but the stability of matter.
\begin{theorem}[Stability of matter]\label{thm:stability}
There exists a constant $C$ depending only on the highest nuclear charge $\bar{Z}$ such that
\begin{equation}
\boxed{\cF_{T,\mu}(\omega,D)\geq -C\big(1+T^{5/2}+\mu_+^{5/2}\big)|D|}
\label{eq:stability-matter} 
\end{equation}
for all $T\geq0$, $\mu\in\R$ and almost all $\omega\in\Omega$.
\end{theorem}

Here $\mu_+=\max(0,\mu)$ denotes the positive part of $\mu$. We see that $\cF_{T,\mu}(\omega,D)$ can be very positive (depending on the number and on the positions of the nuclei), but it can never be too negative. Theorem~\ref{thm:stability} was proved for the first time by Dyson and Lenard in~\cite{DysLen-67,DysLen-68} and it was later revisited by Lieb and Thirring in~\cite{LieThi-75}. For a recent proof, see~\cite{LieSei-09} and~\cite[Theorem 3]{HaiLewSol_2-09}. The only important property of the nuclei used to get this bound is the fact that their charge is uniformly bounded by $\bar{Z}$. The stationarity and the fact that the number of nuclei is locally bounded are not used to get the lower bound~\eqref{eq:stability-matter}.

If the charges of the nuclei are not random but they are all equal to the same charge $Z$, there is a better lower bound on $\cF_{T,\mu}(\omega,D)$:
\begin{equation}
\cF_{T,\mu}(\omega,D)\geq -C\big(1+T^{5/2}+\mu_+^{5/2}\big)|D|+\frac{Z^2}{8}\sum_{(R,z)\in\cK(\omega)\cap D}\frac{1}{\delta_{R,z}}
\label{eq:stability-matter-LY} 
\end{equation}
where we recall that $\delta_{R,z}$ is the distance of $(R,z)$ to the nearest nucleus in $\cK(\omega)$. This bound can be obtained by using an inequality due to Lieb and Yau~\cite{LieYau-88} (generalizing another one of Baxter~\cite{Baxter-80}), see~\cite{LieSei-09} and~\cite[Theorem. 3]{HaiLewSol_2-09}.
In the thermodynamic limit, we hope to prove that $\cF_{T,\mu}(\omega,D)\simeq C|D|$. A necessary condition is at least that 
$$\bE\left(\sum_{(R,z)\in\cK(\omega)\cap D}\frac{1}{\delta_{R,z}}\right)\leq C|D|$$
for $D$ smooth and large enough. By stationarity, this is equivalent to 
\begin{equation}
\bE\big(X_1\big)=\bE\left(\sum_{(R,z)\in\cK(\omega)\cap {W}}\;\frac{1}{\delta_{R,z}}\right)<\ii.
\label{eq:mild_assumption_nuclei} 
\end{equation}
We recall that $X_1$ is defined before in~\eqref{eq:def_X_1}. This
inequality says that inverse of the smallest distance between nuclei in
${W}$ has to be summable in average and this implies that the
number $X_0(\omega)$ of nuclei in ${W}$ (defined in
\eqref{eq:def_X_0}) has to be in $L^1(\Omega)$ as well (see Lemma~\ref{lem:borne-X0} below). All this is already much stronger than saying that the number of nuclei is finite and that the nuclei are distinct almost surely, as we have done before. 

We are able to prove the existence of the thermodynamic limit under a similar but stronger assumption than~\eqref{eq:mild_assumption_nuclei}, see~\eqref{eq:assumption_nuclei}. Before stating our main result, we however need to introduce a notion of regular domains, following~\cite{HaiLewSol_1-09,HaiLewSol_2-09}.

\begin{definition}[Regular domains~\cite{HaiLewSol_1-09,HaiLewSol_2-09}]\label{def_regular}\ \\
$\bullet$ \emph{(Sets with regular boundary)} Let $a>0$ be a real number. We say that a domain $D\subset\R^3$ has \emph{an $a$-regular boundary in the sense of Fisher} if
$$\forall t\in[0,1/a),\qquad \left|\left\{x\in\R^3\ |\ \textnormal{d}(x,\partial D)\leq |D|^{1/3}t\right\}\right|\leq |D|\,a\,t,$$
where $\partial D=\overline{D}\setminus{D}$ is the boundary of $D$.

\medskip

\noindent$\bullet$ \emph{(Cone property)} Let $\epsilon>0$ be a real number. We say that a set $A\subset\R^3$ has the \emph{$\varepsilon$-cone property} if for any $x\in A$ there is a unit vector $v_x\in\R^3$ such that 
$$
\{y\in\R^3\ |\  (x-y)\cdot v_x> (1-\varepsilon^2) |x-y|,\
|x-y|<\varepsilon\}
\subseteq A.
$$ 

\medskip

\noindent$\bullet$ We introduce the set $\cR_{a,\epsilon}$ of all bounded open subsets $D\subset\R^3$ which have an $a$-regular boundary and such that both $D$ and $\R^3\setminus D$ have the $\epsilon$-cone property.
\end{definition}

Our main result is the following.

\begin{theorem}[Existence of the thermodynamic limit for random quantum Coulomb systems]\label{thm:thermo-limit}
We assume that the distribution of nuclei $\cK(\omega)$ is stationary in the sense of~\eqref{eq:stat_K} and that it satisfies
\begin{equation}
\norm{X_1}_{L^p(\Omega)}=\norm{\sum_{(R,z)\in\cK(\omega)\cap {W}}\;\frac{1}{\delta_{R,z}}}_{L^p(\Omega)}<\ii
\label{eq:assumption_nuclei} 
\end{equation}
for some $2\leq p\leq\ii$.
Then there exists a (deterministic) function ${f}(T,\mu)$ such that, for any sequence $(D_n)\subset \cR_{a,\epsilon}$ with $a,\epsilon>0$, $|D_n|\to\ii$ and ${\rm diam}(D_n)|D_n|^{-1/3}\leq c$, we have
\begin{equation}
\lim_{n\to\ii}\bE\left|\frac{\cF_{T,\mu}(\,\cdot\,,D_n)}{|D_n|}-{f}(T,\mu)\right|^q=0
\label{eq:thermo-limit}
\end{equation}
for $q=1$ if $p=2$, and for all $1\leq q<p/2$ if $p>2$. In particular, we have for a subsequence
\begin{equation}
\lim_{n_k\to\ii}\frac{\cF_{T,\mu}(\omega,D_{n_k})}{|D_{n_k}|}={f}(T,\mu)
\label{eq:almost-sure}
\end{equation}
almost surely.
\end{theorem}

We believe that almost-sure convergence as in~\eqref{eq:almost-sure} holds for the \emph{whole sequence $D_n$}, provided that it does not escape too fast to infinity. A simple condition is that $D_n\subset B_{c|D_n|^{1/3}}$ where $B_r$ denotes the ball of radius $r$, centered at $0\in\R^3$. A condition of this type is needed for Birkhoff's almost-sure ergodic theorem (Theorem~\ref{thm:Birkhoff} below). Actually, under this additional assumption, we are able to prove that 
\begin{equation}
\liminf_{n\to\ii}\frac{\cF_{T,\mu}(\omega,D_{n})}{|D_{n}|}={f}(T,\mu),
\label{eq:almost-sure-liminf}
\end{equation}
see Remarks~\ref{rmk:liminf1} and~\ref{rmk:liminf2} below. However, proving that~\eqref{eq:almost-sure-liminf} holds with a limsup may require more involved tools from the theory of probability.

In the next section we will give simple examples of distributions of nuclei satisfying the condition~\eqref{eq:assumption_nuclei}. Recall that when all the charges are equal to $Z$, a condition of the same form as~\eqref{eq:assumption_nuclei} is necessary, with the power $p=1$. Our assumption~\eqref{eq:assumption_nuclei} with $p\geq 2$ on the distribution of nuclei is far from optimal. The power $2$ (which is due to several rough kinetic energy estimates) is probably an artefact of our proof. We have not tried to improve the condition~\eqref{eq:assumption_nuclei} too much. 

It is possible to prove that, in the thermodynamic limit, the system wants to be neutral in average. This can then be used to show that the chemical potential $\mu$ does not play any special role here: $f(T,\mu)$ is just linear with respect to $\mu$.

\begin{corollary}[Asymptotic neutrality and form of $f(T,\mu)$]\label{cor:neutrality}
Under the same assumptions as in Theorem~\ref{thm:thermo-limit}, let 
$$N_{T,\mu}(\omega,D_n):=\frac{\displaystyle\sum_{N\geq0}N\;\tr_{\bigwedge_1^NL^2(\R^3)}e^{-\big(H(\omega,D_n,N)-\mu\,N\big)/T}}{\displaystyle \sum_{N\geq0}\tr_{\bigwedge_1^NL^2(\R^3)}e^{-\big(H(\omega,D_n,N)-\mu\,N\big)/T}}$$ 
be the total average number of electrons in $D_n$. Then we have
\begin{equation}
\lim_{n\to\ii}\bE\left|\frac{N_{T,\mu}(\cdot,D_n)}{|D_n|}-Z_{\rm av}\right|^2=0
\end{equation}
where
\begin{equation}
Z_{\rm av}:=\frac{1}{|W|}\;\bE\left(\sum_{(R,z)\in W\cap\cK(\omega)}z\right)
\end{equation}
is the average nuclear charge per unit cell. Moreover, we have
\begin{equation}
f(T,\mu)=f(T,0)-\mu\, Z_{\rm av}
\end{equation}
and $T\mapsto f(T,0)$ is concave.
\end{corollary}

It is easy to see that $Z_{\rm av}$ is finite under our assumption~\eqref{eq:assumption_nuclei} on $X_1$, see Lemma~\ref{lem:borne-X0} below. Corollary~\ref{cor:neutrality} is proved later in Section~\ref{sec:proof_cor_neutrality}. It easily follows from the stability of matter (Theorem~\ref{thm:stability}) and the upper bound $\bE\big(\cF_{T,\mu}(\cdot,D_n)\big)\leq C|D_n|$ which is proved in Lemma~\ref{lem:upper_bound} below.

\begin{remark}
In this paper we have considered a discrete group acting on $\R^3$ and on the probability space $\Omega$, in an ergodic fashion. Our method of proof can be applied to deal with a \emph{continuous group} like $\R^3$. For instance, the thermodynamic limit exists if the nuclei are distributed using a Poisson process, since the corresponding random variable $X_1$ satisfies~\eqref{eq:assumption_nuclei} in this case.
\end{remark}

\section{The case of i.i.d. perturbations of a nuclear lattice}
\label{sec:nuclei}

In this section we come back to the special (but instructive) case of i.i.d perturbations of nuclei on a lattice, as introduced in Example~\ref{ex:iid}. In this setting we derive conditions under which the assumption~\eqref{eq:assumption_nuclei} on $X_1$ in Theorem~\ref{thm:thermo-limit} is satisfied. For simplicity we consider a cubic lattice $\Z^3$ with only one nucleus per unit cell $W=[-1/2,1/2)^3$. The extension to a more general situation is straightforward.

As in Example~\ref{ex:iid}, we assume that $\Omega=(\Omega_0)^{\Z^3}$, $\bP=(\bP_0)^{\otimes\Z^3}$ and that 
\begin{equation}
\cK(\omega)=\Big\{ \big(j+ r_j(\omega)\,,\,z_j(\omega)\big)\ :\ j\in\Z^3\Big\}.
\label{eq:iid_bis}
\end{equation}
with $r_j(\omega)=r(\omega_j)$ and $z_j(\omega)=z(\omega_j)$. Here $r:\Omega_0\to\R^3$ and $z:\Omega_0\to[\underline{Z},\bar{Z}]$ are two fixed random variables. Let us denote by $\nu$ the law of $r$. We will give conditions on $\nu$ which ensure that $X_1\in L^p(\Omega)$ for some $p\geq2$, as required in Theorem~\ref{thm:thermo-limit}. The simple Gaussian case~\eqref{eq:Gaussian} is covered by the next two results.

We recall that the random variables $X_0$ and $X_1$ are defined by
\eqref{eq:def_X_0} and \eqref{eq:def_X_1}, respectively, that is,
\begin{equation*}
X_0(\omega) = \#\left(\cK(\omega)\cap {W}\right) =
\#\left\{j\in \Z^3\ :\ j+r_j(\omega) \in {W} \right\}=\sum_{j\in\Z^3}\1(j+r_j(\omega)\in W)
\end{equation*}
and
$$X_1(\omega) = \sum_{j\in \Z^3} \frac {\1(j+r_j(\omega)\in {W})} {\delta_j(\omega)},
\quad \text{with}\ \delta_j(\omega) = \inf_{k\in \Z^3\setminus\{j\}} |j+r_j(\omega)
- k -r_k(\omega)|.$$
Recall that $\bP\big(j+r_j\in {W}\big)=\bP_0\big(r\in W-j\big)=\nu(W-j)$, hence
$$\bE\big(X_0\big) = \sum_{j\in \Z^3} \bP(j+r_j\in W)= \sum_{j\in\Z^3} \nu(W-j) = \nu(\R^3)= 1$$
and in particular $X_0\in L^1(\Omega)$. The following elementary proposition deals with the integrability of higher powers of $X_0$.

\begin{proposition}[Integrability of $X_0$ in the i.i.d. case]\label{prop:iid_X_0}
We assume that $\cK$ is of the form~\eqref{eq:iid_bis}. Then
\begin{equation}
\norm{X_0}_{L^p(\Omega)}\leq \sum_{j\in\Z^3}\nu\big(W-j\big)^{1/p}.
\label{eq:estim_X_0_L_p} 
\end{equation}
Thus $X_0\in L^p(\Omega)$ when the right side is finite. If the support of $\nu$ is \emph{not} compact, then $X_0\notin L^\ii(\Omega)$.
\end{proposition}

\begin{proof}  
Using that $X_0=\sum_{j\in\Z^3}\1\big(j+r_j\in {W}\big)$, we obtain
$$\norm{X_0}_{L^p(\Omega)}\leq \sum_{j\in\Z^3}\norm{\1\big(j+r_j\in {W}\big)}_{L^p(\Omega)}=\sum_{j\in\Z^3}\bP\big(j+r_j\in{W}\big)^{1/p}=\sum_{j\in\Z^3}\nu\big(W-j\big)^{1/p}.$$
This bound does not really use the independence of the variables $r_i$. We now prove that $X_0\notin L^\ii(\Omega)$ when the support of $\nu$ is not compact. This means that there exists an infinite sequence $(j_n)\subset\Z^3$ such that 
$$\forall n,\qquad \bP(r_{j_n} \in {W}-j_n)=\bP_0(r \in {W}-j_n)=\nu(W-j_n)>0.$$
Let $N$ be a positive integer. It is clear that if for all $1\leq n\leq N$ we have $j_n+r_{j_n} \in {W}$, then $X_0\geq N$. Hence, we
have
$$\bP\left(X_0 \geq N\right) \geq \bP\left(\bigcap_{n=1}^N
  \left\{j_n+ r_{j_n} \in  {W}\right\} \right) = \prod_{n=1}^N \bP(j_n+r_{j_n} \in
{W}).$$
Since  $\bP(j_n+r_{j_n} \in {W}) >0$ for all $n$,
we deduce that $\bP(X_0\geq N) >0$ for any $N>0$. This implies that $X_0\not\in L^\ii(\Omega)$.
\end{proof}

The following is now a simple application of the previous proposition.

\begin{example}[Gaussian perturbations I]
Assume that $\nu(x)=(2\pi \sigma)^{-3/2}e^{-{|x|^2}/(2\sigma)}$ is a Gaussian distribution. Then $X_0\in L^p(\Omega)$ for all $1\leq p<\ii$ but $X_0\notin L^\ii(\Omega)$.
\end{example}

\medskip

We next turn to the study of $X_1$. The simpler situation is when the nuclei never escape from their cell and stay at a finite distance to the other ones.
\begin{lemma}[Small perturbations]
If $\nu$ has its support inside $W$, then both $X_0$ and $X_1$ are in $L^\ii(\Omega)$.
\end{lemma}
\begin{proof}
Under the assumption, it is obvious that there is always exactly one nucleus per unit cell, and that it is at a finite distance $\eta$ to any other nucleus. Thus $X_0\equiv1$ and $X_1\leq 1/\eta$.
\end{proof}

\begin{corollary}[Thermodynamic limit for small i.i.d. perturbations of the nuclei]
We assume that $\cK$ is of the form~\eqref{eq:iid_bis} and that $\nu$ has its support inside $W$. 
Then the thermodynamic limit in Theorem~\ref{thm:thermo-limit} 
\begin{equation}
\lim_{n\to\ii}\bE\left|\frac{\cF_{T,\mu}(\,\cdot\,,D_n)}{|D_n|}-{f}(T,\mu)\right|^q=0
\label{eq:thermo-limit-bis}
\end{equation}
is valid for any $1\leq q<\ii$.
\end{corollary}

A more interesting situation is covered in the following

\begin{proposition}[Integrability of $X_1$ in the i.i.d. case]\label{prop:iid_X_1}
We assume that $\cK$ is of the form~\eqref{eq:iid_bis} and that $\nu$ satisfies 
\begin{equation}
\sum_{j\neq0}\norm{\nu}_{L^\ii(W+B_\eta-j)}^{1/p}<\ii
\label{eq:cond_nu_X_1}
\end{equation}
for some $\eta>0$ and some $1\leq p<3$. Then we have 
$X_1\in L^p(\Omega)$.

If moreover there exist a ball $B_\kappa(v)\subset W$ with radius $\kappa>0$ and center $v\in W$, and $i\neq j\in\Z^3$ such that 
\begin{equation}
\forall x\in \big(B_\kappa(v)-i\big)\cup\big(B_\kappa(v)-j\big),\quad \nu(x)\geq \kappa>0,
\label{eq:cond_nu_X_1bis}
\end{equation}
then $X_1\notin L^3(\Omega)$.
\end{proposition}

\medskip

The condition~\eqref{eq:cond_nu_X_1} implies that the measure $\nu$ is actually a bounded function outside of the unit cell $W$ (indeed, outside of $(1-\eta/\sqrt{3})W$), and that it decays fast enough at infinity, in a similar fashion as in~\eqref{eq:estim_X_0_L_p}. On the other hand, $\nu$ does not have to be absolutely continuous with respect to the Lebesgue measure inside $W$. The condition~\eqref{eq:cond_nu_X_1bis} means that, with a positive probability, two nuclei coming from different sites $i$ and $j$ will both (independently) visit all of the same ball $B_\kappa(v)$.

\begin{proof}
Recall that 
$$X_1 = \sum_{i\in\Z^3} \1_{{W}}(j+r_j) \frac 1 {\delta_j},$$
therefore 
\begin{equation}
\norm{X_1}_{L^p(\Omega)} \leq \sum_{j\in\Z^3} \norm{\1_{{W}}(j+r_j) \frac 1 {\delta_j}}_{L^p(\Omega)}=\sum_{j\in\Z^3} \norm{\1_{{W-j}}(r_0) \frac 1 {\delta_0}}_{L^p(\Omega)}
\label{eq:estim_X_1} 
\end{equation}
for $p\geq1$. 
In order to show that the series on the right side is convergent, we first estimate $\bP(r_0\in W-j\ \cap\ \delta_0 < \epsilon)$. 
For this purpose, we point out that $\delta_0<\epsilon$ if and only if there exists $k\neq 0$ such that $|r_0 - k
  -r_k|<\epsilon$. Hence, 
\begin{align*}
\bP(r_0\in W-j\ \cap\ \delta_0 < \epsilon) &\leq \sum_{k\neq 0} \bP\left(r_0\in W-j\ \cap\  |r_0 - k
  -r_k|<\epsilon\right) \\
&=  \sum_{k\neq 0}\int_{W-j}\nu(x)\,dx\int_{\R^3}\nu(y)\,dy\,
\1_{B_\epsilon} (x -k -y) \\
&=\sum_{k\neq 0}\int_{W}\nu(x-j)\,dx\int_{B_\epsilon}\nu(y+x-j-k)\,dy.
\end{align*}
Even if $\nu$ is not necessarily absolutely continuous in the interior of $W$, we have used an integral notation for simplicity.
When $j=0$, we obtain, for $\epsilon < \eta,$ 
\begin{align*}
\bP(r_0\in W\ \cap\ \delta_0 < \epsilon) &\leq\sum_{k\neq 0}\int_{W}\nu(x)\,dx\int_{B_\epsilon}\nu(y+x-k)\,dy\\
&\leq  \left(\sum_{k\neq 0}\norm{\nu}_{L^\ii(W+B_\eta-k)}\right)\,\nu(W)\,(4\pi/3)\epsilon^3=C\epsilon^3\;\nu(W).
\end{align*}
In the second line we have used that $\nu(y+x-k)\leq\norm{\nu}_{L^\ii(W+B_\eta-k)}$ for $\epsilon<\eta$. Note that
$$\sum_{j\neq0}\norm{\nu}_{L^\ii(W+j+B_\eta)}^{1/p}<\ii\ \Longrightarrow\ \sum_{j\neq0}\norm{\nu}_{L^\ii(W+j+B_\eta)}<\ii$$
since $p\geq1$.
When $j\neq0$, we isolate the term $k=-j$ and obtain 
\begin{multline*}
\sum_{k\neq 0}\int_{W}\nu(x-j)\,dx\int_{B_\epsilon}\nu(y+x-j-k)\,dy\\
=\sum_{\substack{k\neq 0\\ k\neq -j}}\int_{W}\nu(x-j)\,dx\int_{B_\epsilon}\nu(y+x-j-k)\,dy+\int_{W}\nu(x-j)\,dx\int_{B_\epsilon}\nu(y+x)\,dy.
\end{multline*}
We estimate the first term by (recall that $\epsilon<\eta$)
\begin{multline*}
\sum_{\substack{k\neq 0\\ k\neq -j}}\int_{W}\nu(x-j)\,dx\int_{B_\epsilon}\nu(y+x-j-k)\,dy\\
\leq \left(\sum_{k\neq -j}\norm{\nu}_{L^\ii(W+B_\eta-j-k)}\right)\nu(W-j)(4\pi/3)\epsilon^3\leq C\,\norm{\nu}_{L^\ii(W-j)}\epsilon^3,
\end{multline*}
where we have used $\nu(W-j)\leq\norm{\nu}_{L^\ii(W-j)}$. For the second term we write
\begin{align*}
\int_{W}\nu(x-j)\,dx\int_{B_\epsilon}\nu(y+x)\,dy&\leq \norm{\nu}_{L^\ii(W-j)}\int_{W}\,dx\int_{B_\epsilon}\nu(y+x)\,dy\\
&\leq \norm{\nu}_{L^\ii(W-j)}\nu(W+B_\eta)(4\pi/3)\epsilon^3\\
&=C\,\norm{\nu}_{L^\ii(W-j)}\epsilon^3.
\end{align*}
As a conclusion, we have shown that
$$\bP(r_0\in W-j\ \cap\ \delta_0 < \epsilon)\leq C\epsilon^3\begin{cases}
\nu(W)&\text{for $j=0$,}\\
\norm{\nu}_{L^\ii(W-j)} & \text{for $j\neq0$.}
\end{cases}$$

We then compute the expectation value of $\1(r_0\in W-j)\ \delta_0^{-q}$:
\begin{align*}
&  \bE\left(\frac{\1(r_0\in W-j)}{\delta_0^{q}}\right)\\
&\qquad = \bE\left(\frac{\1(r_0\in W-j)}{\delta_0^q}\1_{\delta_0\geq 2} \right) +\sum_{k\geq 0}\bE\left( \frac{\1(r_0\in W-j)}
    {\delta_0^q}\1_{2^{-k}\leq \delta_0< 2^{-(k-1)}}\right) \\
&\qquad\leq \frac {\nu(W-j)} {2^q} + \sum_{k\geq 0} 2^{qk} \bP\left(r_0\in W-j\ \cap\ \delta_0 <
  2^{-(k-1)}\right)\\
&\qquad \leq  \frac {\nu(W-j)} {2^q} + C\Big(\1(j=0)\nu(W)+\1(j\neq0)\norm{\nu}_{L^\ii(W-j)}\Big)\sum_{k\geq 0} 2^{qk}2^{-3(k-1)}.
\end{align*}
The sum is convergent provided that $1\leq q<3$. Hence we have shown that
$$\forall 1\leq q<3,\qquad \norm{\frac{\1(r_0\in W-j)}{\delta_0}}_{L^q(\Omega)}\leq \frac{C}{(1-2^{q-3})^{1/q}}
\begin{cases}
\nu(W)^{1/q}&\text{for $j=0$,}\\
\norm{\nu}_{L^\ii(W-j)}^{1/q} & \text{for $j\neq0$.}
\end{cases}
$$
Inserting this in~\eqref{eq:estim_X_1} gives the result under our
assumption~\eqref{eq:cond_nu_X_1}. 

In order to conclude the proof of Proposition~\ref{prop:iid_X_1}, we
show that $X_1\not\in L^3$ if \eqref{eq:cond_nu_X_1bis} holds. For this purpose, we write
$$X_1^3 = \sum_{\substack{i\in \Z^3 \\ i+r_i\in W}}
\sum_{\substack{j\in \Z^3 \\ j+r_j\in W}} \sum_{\substack{k\in
    \Z^3 \\ k+r_k\in  W}} \frac 1 {\delta_i \delta_j \delta_k}
\geq  \sum_{\substack{i\in \Z^3 \\ i+r_i\in W}} \frac 1
{\delta_i^3}.$$
Therefore it is sufficient to prove that $\delta_i^{-3}\1(i+r_i\in W)\not\in
L^1(\Omega)$ for some $i$. We choose $i,j$ such that
\eqref{eq:cond_nu_X_1bis} is satisfied. We have in particular
$\bP(i+r_i\in W) = \nu(W-i) >0.$
We compute
\begin{align*}
  \bP\left(\delta_i>\epsilon\; \cap\; i+r_i\in W\right) &= \bP\left(\left\{i+r_i\in
      W \right\}\cap \bigcap_{k\neq i} \left\{|i+r_i -k - r_k|>\epsilon \right\}\right) \\
&= \lim_{N\to\infty} \bP\left(\left\{i+r_i\in
      W \right\}\cap \bigcap_{\substack{|k|\leq N \\ k\neq i}} \left\{|i+r_i -k - r_k|>\epsilon \right\} \right).
\end{align*}
Using $1-t \leq e^{-t}$, we obtain the bound 
\begin{align*}
&\bP \left(\{i+r_i\in {W}\}\cap\bigcap_{\substack{|k|\leq N\\ k\neq i}} \left\{
        |i+r_i - k - r_k|>\epsilon\right\}\right)\\
&\qquad\qquad= \int_{{W}-i}\nu(y_i)\,dy_i \prod_{\substack{|k|\leq N\\ k\neq i}}\int_{\R^3}\nu(y_k)\,dy_k  \1\big(|i+y_i-k-y_k|>\epsilon\big)\\
&\qquad\qquad= \int_{{W}-i}\nu(y_i)\,dy_i \prod_{\substack{|k|\leq N\\ k\neq i}}\left(1-\int_{B_\epsilon}\nu(y_k+k-y_i-i)\,dy_j  \right)\\
&\qquad\qquad\leq \int_{{W}-i}\nu(y_i) \exp\left[-\displaystyle\int_{B_\epsilon}\left(\sum_{\substack{|k|\leq N\\ k\neq i}}\nu(y+k-y_i-i)\right)\,dy\right]\,dy_i.
\end{align*}
Passing to the limit $N\to\ii$, we have shown that
$$\bP\big(\delta_i >\epsilon\ \cap\ i+r_i\in {W}\big)\leq \int_{{W}-i}\nu(y_i) \exp\left[-\displaystyle\int_{B_\epsilon}\left(\sum_{k\neq 0}\nu(y+k-y_i)\right)\,dy\right]\,dy_i.$$
When $\epsilon<\kappa/2$ we can use our assumption~\eqref{eq:cond_nu_X_1bis} and infer that
$\nu(y+k-y_i)>\kappa$ for $y_i\in B_{\kappa/2}(v)-i$, $y\in B_\epsilon(0)$ and $k=j-i$. In particular,
\begin{multline*}
\bP\big(\delta_i >\epsilon\ \cap\ i+r_i\in {W}\big)\\
\leq \int_{{W}\setminus B_\kappa(v)-i}\nu(y_i) \,dy_i+\int_{B_\kappa(v)-i}\nu(y_i) \exp\left[-C\kappa\epsilon^3\right]\,dy_i
\leq (1-C\epsilon^3)\int_{W-i}\nu(y_i) \,dy_i
\end{multline*}
where $C>0$ depends on $i$ (a fixed index). We deduce that  
$$\bP\big(\delta_i \leq\epsilon\ \cap\ i+r_i\in {W}\big)\geq C\epsilon^3\,\bP\big(i+r_i\in {W}\big)$$
and finally obtain 
$$\bE\left(\frac{\1(i+r_i\in {W})}{\delta_i^3}\1(\delta_i\leq \epsilon)\right)\geq C\,\bP\big(i+r_i\in {W}\big).$$
If $\1(i+r_i\in {W})\delta_i^{-3}$ were in $L^1(\Omega)$, the left side
would converge to 0 when $\epsilon\to0$, by the dominated convergence
theorem. Since the right side is $>0$ by our choice of $i$, and
independent of $\epsilon$, we deduce that $\1(i+r_i\in
{W})\delta_i^{-3}\notin L^1(\Omega)$. This concludes the proof of Proposition~\ref{prop:iid_X_1}.
\end{proof}

Let us recall the inequality~\eqref{eq:stability-matter-LY} which implies that, when $z(\omega)=Z$ a.s.,
$$\frac{Z^2}8 \sum_{\substack{j\ :\\ W-j\subset D}}X_1(\tau_j\omega) \leq \cF_{T,\mu}(\omega,D) + C\left(1+T^{5/2} +
  \mu_+^{5/2}\right)|D|.$$
This clearly shows that $\cF_{T,\mu}(\omega,D)$ is not in $L^3(\Omega)$ when~\eqref{eq:cond_nu_X_1bis} is satisfied.

\begin{corollary}[Thermodynamic limit for large i.i.d. perturbations of the nuclei]
We assume that $\cK$ is of the form~\eqref{eq:iid_bis} and that $\nu$ satisfies 
\begin{equation}
\sum_{j\neq0}\norm{\nu}_{L^\ii(W+B_\eta-j)}^{1/3}<\ii
\label{eq:cond_nu_X_1/3}
\end{equation}
for some $\eta>0$. Then the thermodynamic limit in Theorem~\ref{thm:thermo-limit} 
\begin{equation}
\lim_{n\to\ii}\bE\left|\frac{\cF_{T,\mu}(\,\cdot\,,D_n)}{|D_n|}-{f}(T,\mu)\right|^q=0
\label{eq:thermo-limit-ter}
\end{equation}
is valid for all $1\leq q<3/2$.
If moreover $z(\omega)=Z$ a.s. and~\eqref{eq:cond_nu_X_1bis} is satisfied, then $\cF_{T,\mu}(\cdot,D_n) \not \in L^3(\Omega)$ for any $n$. Thus~\eqref{eq:thermo-limit-ter} cannot hold with $q=3$.
\end{corollary}

\begin{example}[Gaussian perturbations II]
Assume that $\nu(x)=(2\pi \sigma)^{-3/2}e^{-{|x|^2}/(2\sigma)}$ is a Gaussian distribution. Then we have for all $j\in\Z^3$
$$\norm{e^{-{|x|^2}/(2\sigma)}}_{L^\ii(W-j)}\leq Ce^{-|j|^2/(4\sigma)}.$$
Since~\eqref{eq:cond_nu_X_1} is satisfied for all $1\leq p<3$, we have $X_1\in L^p(\Omega)$ for all $1\leq p<3$. On the other hand, the support of $\nu$ is the whole space, hence~\eqref{eq:cond_nu_X_1bis} is obviously verified, thus $X_1\notin L^3(\Omega)$. The thermodynamic limit exists in $L^{q}(\Omega)$ for all $1\leq q<3/2$. But $\cF_{T,\mu}(\cdot,D) \not \in L^3(\Omega)$ when  $z(\omega)=Z$ a.s..
\end{example}

In the next section we provide the detailed proof of Theorem~\ref{thm:thermo-limit}.

\section{Proof of Theorem~\ref{thm:thermo-limit}}

Our proof follows the technique introduced in~\cite{HaiLewSol_1-09,HaiLewSol_2-09}. In~\cite{HaiLewSol_1-09}, abstract conditions called $\text{(A1)--(A6)}$ ensuring the existence of the thermodynamic limit for a functional $D\mapsto \cF(D)$ were provided. These conditions were verified in~\cite{HaiLewSol_2-09} for the deterministic crystal as well as some other quantum systems. Our technique of proof for the stochastic case can be sketched as follows:

\medskip

\noindent $(i)$ We start by proving in Lemma~\ref{lem:upper_bound} below that $|D_n|^{-1}\cF_{T,\mu}(\omega,D_n)$ is uniformly bounded in $L^{p/2}(\Omega)$. This uses an adequate trial state together with the assumption~\eqref{eq:assumption_nuclei} on the distribution of nuclei.

\medskip

\noindent $(ii)$ Then, we show that the averaged free energy $\bE\,\big(\cF_{T,\mu}(\cdot,D)\big)$ satisfies all the abstract properties $\text{(A1)--(A6)}$ of~\cite{HaiLewSol_1-09,HaiLewSol_2-09}. Hence its thermodynamic limit exists and we call $f(T,\mu)$ the corresponding limit. Note that $\bE\,\big(\cF(\cdot,D)\big)$ is \emph{periodic} by stationarity of $\cF_{T,\mu}$, hence the formalism of~\cite{HaiLewSol_1-09,HaiLewSol_2-09} is appropriate. This step requires some upper bounds in average (in particular the estimate $\text{(A4)}$ proved in Lemma~\ref{lem:A4} below) in which the assumption~\eqref{eq:assumption_nuclei} on the distribution of nuclei is again used.

\medskip

\noindent $(iii)$ We show that
\begin{equation}
\lim_{n\to\ii}\bE\left[\frac{\cF_{T,\mu}(\omega,D_n)}{|D_n|} - f(T,\mu)\right]_-=0
\label{eq:liminf} 
\end{equation}
where $[x]_-=\max(0,-x)$. This step uses the Graf-Schenker-type
inequality $\text{(A5)}$ (Lemma~\ref{lem:Graf-Schenker} below) which is
a precise lower bound on $\cF_{T,\mu}(\omega,D)$ at fixed $\omega$, in
terms of a tiling of simplices. The limit~\eqref{eq:liminf} is obtained
by a suitable application of the ergodic theorem (Theorem~\ref{thm:Birkhoff} below).

\medskip

\noindent $(iv)$ A simple argument shows that the convergence of the average $\bE\,\big(\cF_{T,\mu}(\cdot,D)\big)$ and~\eqref{eq:liminf} imply the strong convergence in $L^1(\Omega)$. By interpolation, the convergence in $L^q(\Omega)$ for $1\leq q<p/2$ follows.

\bigskip

Actually in our proof we do not consider the original functional $\cF_{T,\mu}(\omega,D)$ but, like in~\cite{HaiLewSol_2-09}, we optimize over the charges of the nuclei which are close to the boundary of $D$. This provides a modified functional $\underline\cF_{T,\mu}(\omega,D)$ to which the previous scheme is applied. Only in the end of our proof we come back to the original free energy.

\subsection*{Step 1. Bound on $X_0$}
A preliminary result is the following:
\begin{lemma}[Bounds on $X_1$ give bounds on $X_0$]
  \label{lem:borne-X0}
Assume that the distribution of nuclei ${\cK}(\omega)$ is stationary
in the sense of \eqref{eq:stat_K}. Then, for any $p\geq 1$, $X_1\in
L^p(\Omega)$ implies $X_0\in L^p(\Omega)$.
\end{lemma}
\begin{proof}[Proof of Lemma~\ref{lem:borne-X0}]
Recall that $X_0\geq 0,$ so we have
\begin{equation}
\label{eq:separation}
\bE\left(|X_0|^p\right) = \bE\left(X_0^p\right) = \bE\left(X_0^p
  {\1}_{X_0 \leq 1}\right) +\bE\left(X_0^p
  {\1}_{X_0 > 1}\right).
\end{equation}
Next, we point out that if $X_0>1$, then, for any $R\in {\cK}\cap
{W},$ we have $\delta_{R,z} \leq \operatorname{diam}(W).$ Hence, 
$$X_1 \1_{X_0 > 1} \geq \frac{X_0}{\operatorname{diam}(W)} \1_{X_0>1}.$$
We insert this estimate into \eqref{eq:separation}, finding
$\bE\left(|X_0|^p\right) \leq 1 + \operatorname{diam}(W)^p
\bE(X_1^p)$,
which concludes the proof.
\end{proof}
\subsection*{Step 2. Upper bounds}
In this first step we will establish some upper bounds in average, that is for $\norm{\cF_{T,\mu}(\cdot,D)}_{L^q(\Omega)}$. It is for these uper bounds that we will need the assumption~\eqref{eq:assumption_nuclei} which gives estimates (in average) on the number of nuclei per unit volume, as well as on the smallest distance between them. The lower bounds will on the contrary be almost uniform in $\omega\in\Omega$ (up to a small error term which is easily controlled, see Lemma~\ref{lem:Graf-Schenker} and Remark~\ref{rmk:estim_error_GS} below). The proofs of the results in this first step are rather technical but the strategy is similar to that used in the deterministic case in~\cite{HaiLewSol_1-09}.

Our first result will be that, under our assumption~\eqref{eq:assumption_nuclei} on the nuclei, the free energy is bounded above by a constant times the volume $|D|$, in average. The following is the random equivalent to~\cite[Prop. 2]{HaiLewSol_1-09}.

\begin{lemma}[Upper bound]\label{lem:upper_bound}
Under the hypotheses of Theorem~\ref{thm:thermo-limit}, we have, for any regular domain $D\in\cR_{a,\epsilon}$
\begin{equation}
\bE\;\Big|\cF_{T,\mu}(\cdot,D)\Big|^{p/2}\leq C\,|D|^{p/2}
\label{eq:simple_upper_bd}
\end{equation}
where we recall that $p$ appears in~\eqref{eq:assumption_nuclei}, and where $C$ depends on $a>0$, $\epsilon>0$, $T\geq0,$ $p$ and $\mu\in\R$, but not on $D\in\cR_{a,\epsilon}$. 

In particular, if we denote by $\Gamma(\omega)$ the electronic density matrix in Fock space $\mathscr{F}$ of any optimal state for $\cF_{T,\mu}(\omega,D)$, we have the bound in average for the electronic density $\rho_\Gamma$ and the kinetic energy
\begin{equation}
\bE\;\left(\int_{D}\rho_\Gamma+\int_{D}\rho_\Gamma^{5/3}+\tr_{\mathscr{F}}\Big(\sum_i (-\Delta)_i\;\Gamma\Big)\right)^{q}\leq C\,|D|^{q}
\label{eq:upper-bd-average}
\end{equation}
for all $1\leq q\leq p/2$.
\end{lemma}

\medskip

\begin{proof}[Proof of Lemma~\ref{lem:upper_bound}]
Our proof follows the one of~\cite[Prop. 2]{HaiLewSol_2-09} where some missing details can be found. There, the nuclei were all assumed to have a finite distance to their nearest neighbor. Our task here is to exhibit the dependence of this upper bound in terms of these parameters.

We only write the proof for $\mu=T=0$, the general case being similar~\cite{HaiLewSol_2-09}, and we denote $\cF=\cF_{0,0}$.
Let $D$ be a regular domain and $\omega\in\Omega$. Recall that we have stability of matter which tells us that $\cF(\omega,D)\geq -C|D|$ almost surely. We therefore only have to prove an upper bound, which is done by constructing an appropriate trial state. For each nucleus $(R,z)\in \cK(\omega)\cap D$, we place in $D$ radial electrons of total charge $z$, in a small ball of radius 
$\delta'_{R,z}/8$ where we have defined for convenience
$$\delta'_{R,z}:=\min\big(\delta_{R,z},\epsilon\big).$$ 
Recall that $\delta_{R_,z}(\omega)$ is the distance of the nucleus $(R,z)$ to the closest nucleus in the system and that $\epsilon$ quantifies the cone property of the set $D$. We want to put this ball as close to the nucleus as possible. When the nucleus is at a distance $>\epsilon$ to the boundary $\partial D$, we can simply put the radial electrons on top of the nucleus, leading to a vanishing Coulomb potential outside of the support of the electrons, by Newton's theorem. When the nucleus is at a distance $\leq\epsilon$ to the boundary of $D$, we use the cone property and place the ball at a distance $\delta'_{R,z}/4$ to the nucleus, in the small cone which is inside $D$. This construction is the same as in~\cite{HaiLewSol_2-09} except that our electrons live in small balls depending on $\delta'_{R,z}$. In~\cite{HaiLewSol_2-09} they were all living in balls of constant radius $\epsilon/8$.

To simplify our estimate, we use the notation 
$$(\partial D)_\epsilon:=\{x\in\R^3\ :\ {\rm d}(x,\partial D)\leq \epsilon\}.$$
Since $D$ has an $a$-regular boundary, we have $|(\partial D)_\epsilon|\leq a\epsilon\,|D|^{2/3}$.

The total energy of our trial state contains several terms. The kinetic energy used to squeeze the electrons in their small balls can be estimated by a constant times
$$\sum_{(R,z)\in\cK(\omega)\cap D}\frac{z^{5/3}}{(\delta'_{R,z})^2}.$$
Here we get a coefficient $z^{5/3}$ because of the Pauli principle for the electrons. We first have to put $z$ electrons in a ball of radius 1. For this we just fill in the first eigenvalues of the Dirichlet Laplacian of the unit ball and then average over rotations to make our state radial. Then we scale these electrons to make them fit in a ball of radius $\delta'_{R,z}$. Here the term $z^{5/3}$ is not a problem since the charges $z$ are uniformly bounded by assumption, $z\leq\overline{Z}$. But later this difficulty will pop up again.

The only other term is the interaction between all the charges in $(\partial D)_\epsilon$. The interaction between each nucleus and its screening electronic cloud is negative and we can discard it for an upper bound. For later purposes, we however note that it can be estimated by a constant times
\begin{equation}
\sum_{(R,z)\in\cK(\omega)\cap (\partial D)_\epsilon}\frac{z^2}{\delta'_{R,z}}.
\end{equation}
We are left with the dipole-dipole interactions, which we denote by 
$${\rm Dip}(R,R')=\frac{zz'}{|R-R'|}+\frac{zz'}{|X-X'|}-\frac{zz'}{|X-R'|}-\frac{zz'}{|X'-R|}$$
(even if it also depends on $\omega$). Here $X$ and $X'$ are the positions of the electrons which are such that $|R-X|\leq \delta'_{R,z}/4$ and $|R'-X'|\leq \delta'_{R',z'}/4$. When $|R-R'|$ is sufficiently large, this interaction behaves like $|R-R'|^{-3}$. When $|R-R'|$ is small, we use that the electrons are at a small distance to the nuclei. For instance
$$|R-X'|\geq |R-R'|-|R'-X'|= |R-R'|-\delta_{R',z'}'/4\geq \frac34|R-R'|.$$
Similarly, $|X-X'|\geq |R-R'|/2$. All in all, we deduce that 
\begin{equation}
\big|{\rm Dip}(R,R')\big| \leq \frac{Czz'}{|R-R'|\big(1+|R-R'|^2\big)}.
\label{eq:estim_dipoles_distance} 
\end{equation}
Our final bound on the energy is therefore of the form
\begin{multline}
\cF(\omega,D)\\
\leq C\left(\sum_{(R,z)\in\cK(\omega)\cap D}\frac{z^{5/3}}{(\delta'_{R,z})^2}+\sum_{\substack{(R,z),\,(R',z')\in\cK(\omega)\cap (\partial D)_\epsilon\\ R\neq R'}}\frac{z\,z'}{|R-R'|(1+|R-R'|^2)}\right) 
\label{eq:upper_bound_1}
\end{multline}

In order to simplify our reasoning, we now cover $D$ and $(\partial
D)_\epsilon$ by translations of the domain $W$. This means we write
$D\subset\cup_{j\in\cJ}W_j$ and $(\partial
D)_\epsilon\subset\cup_{j\in\partial\cJ}W_j$ where $\cJ\subset\ccL$ and
$\partial\cJ\subset\ccL$ are such that $\#\cJ\leq C|D|$ and
$\#\partial\cJ\leq C|D|^{2/3}$, by the regularity of $D$ (Figure~\ref{fig:tiling}). We also use
the notation $W_j:=W-j$. 
\begin{figure}
\includegraphics{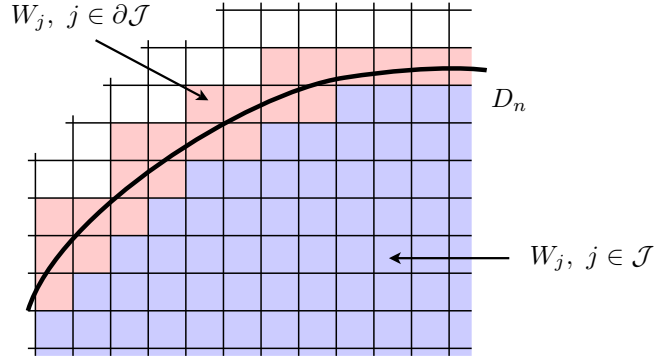}
\caption{The sets of indices $\cJ$ and $\partial\cJ$ in the simple example of a cubic lattice.\label{fig:tiling}}
\end{figure}
Then, the first term
of~\eqref{eq:upper_bound_1} can be estimated by
$$\bar{Z}^{5/3}\sum_{(R,z)\in\cK(\omega)\cap D}\frac{1}{(\delta'_{R,z})^2}\leq \bar{Z}^{5/3}\sum_{j\in\cJ}X'_{2}(\tau_j\omega)\leq C\frac{|D|}{\#\cJ}\sum_{j\in\cJ}X'_{2}(\tau_j\omega)$$
by stationarity, and with the definition
\begin{equation}
X'_{p}(\omega):=\sum_{(R,z)\in\cK(\omega)\cap {W}}\frac{1}{\delta'_{R,z}(\omega)^{p}}.
\label{eq:def_Delta_p} 
\end{equation}
Note that, since $(\delta'_{R,z})^{-1} \leq \epsilon^{-1} +
(\delta_{R,z})^{-1},$
\begin{multline*}
X'_{2}(\omega)=\sum_{(R,z)\in\cK(\omega)\cap
  {W}}\frac{1}{(\delta'_{R,z})^2} \leq
\left(\sum_{(R,z)\in\cK(\omega)\cap
    {W}}\frac{1}{\delta'_{R,z}}\right)^2 \\
\leq  \left(\frac 1 \epsilon\sum_{(R,z)\in\cK(\omega)\cap
    {W}}1 + \sum_{(R,z)\in\cK(\omega)\cap
    {W}}\frac{1}{\delta'_{R,z}}\right)^2 
 =  \left(\frac{X_0}{\epsilon} + X_1\right)^2.
\end{multline*}
This shows that $X'_{2}\in L^{p/2}(\Omega)$, under our
assumption~\eqref{eq:assumption_nuclei} and using Lemma~\ref{lem:borne-X0}.
By the triangular inequality, we get the estimate
\begin{equation*}
\norm{\sum_{(R,z)\in\cK(\omega)\cap D}\frac{z}{(\delta'_{R,z})^2}}_{L^{p/2}(\Omega)}\leq C|D|\,\norm{X'_{2}}_{L^{p/2}(\Omega)}
\end{equation*}
Note that we implicitly use here that $p\geq2$.

We now claim that the second term of~\eqref{eq:upper_bound_1} can be estimated as follows
\begin{multline}
\norm{\sum_{\substack{(R,z),\,(R',z')\in\cK(\omega)\cap (\partial D)_\epsilon\\ R\neq R'}}\frac{z\,z'}{|R-R'|(1+|R-R'|^2)}}_{L^{p/2}(\Omega)}\\
\leq C|D|^{2/3}\left(\norm{X_1}^2_{L^{p}(\Omega)}+\log(1+|D|)\norm{X_0}_{L^{p}(\Omega)}^2\right).
\label{eq:estim_interaction_dipoles}
\end{multline}
When $R$ and $R'$ belong to two adjacent domains $W_j$ and $W_k$, or to the same domain $W_j=W_k$, we only use that 
$$\frac{1}{|R-R'|}\leq \frac12\left(\frac{1}{\delta_{R,z}}+\frac{1}{\delta_{R',z'}}\right)$$
by definition of $\delta_{R,z}$.
When $R$ and $R'$ belong to two domains $W_j$ and $W_k$ which are separated by a finite distance, we use the estimate 
$$\sum_{(R,z)\in\cK\cap W_j}\sum_{(R',z')\in\cK\cap W_k}\frac{z\,z'}{|R-R'|(1+|R-R'|^2)}\leq C\frac{X_{0}(\tau_j\omega)\,X_{0}(\tau_k\omega)}{1+|j-k|^3}$$
where we recall that $X_{0}(\omega)$ is the total number of nuclei in the unit cell. The final estimate on the second term of~\eqref{eq:upper_bound_1} is 
\begin{equation}
C\sum_{\substack{j,k\in\partial\cJ\\ |j-k|\leq C}}X_{1}(\tau_j\omega)X_0(\tau_k\omega)+C\sum_{\substack{j,k\in\partial\cJ\\ j\neq k}}\frac{X_{0}(\tau_j\omega)\,X_{0}(\tau_k\omega)}{1+|j-k|^3}.
\label{eq:upper_bound_2nd_term} 
\end{equation}
We use that 
$$X_{0/1}(\tau_j\omega)\,X_{0/1}(\tau_k\omega)\leq\frac12\Big(X_{0/1}(\tau_j\omega)^2+X_{0/1}(\tau_k\omega)^2\Big)$$
and obtain
$$\norm{\sum_{\substack{j,k\in\partial\cJ\\ |j-k|\leq C}}X_{1}(\tau_j\omega)X_0(\tau_k\omega)}_{L^{p/2}(\Omega)}\leq C|D|^{2/3}\left(\norm{X_{1}}_{L^{p}(\Omega)}^2+\norm{X_{0}}_{L^{p}(\Omega)}^2\right).$$
Similarly,
\begin{multline*}
\sum_{\substack{j,k\in\partial\cJ\\ j\neq k}}\frac{X_{0}(\tau_j\omega)\,X_{0}(\tau_k\omega)}{1+|j-k|^3}\leq \sum_{\substack{j,k\in\partial\cJ\\ j\neq k}}\frac{X_{0}(\tau_j\omega)^2}{1+|j-k|^3}
\leq C\log(\#\partial\cJ)\sum_{j\in\partial\cJ}X_{0}(\tau_j\omega)^2\\
\leq C|D|^{2/3}\log(1+|D|)\left(\frac{1}{\#\partial\cJ}\sum_{j\in\partial\cJ}X_{0}(\tau_j\omega)^2\right).
\end{multline*}
Therefore
$$\norm{\sum_{\substack{j,k\in\partial\cJ\\ j\neq k}}\frac{X_{0}(\tau_j\omega)\,X_{0}(\tau_k\omega)}{1+|j-k|^3}}_{L^{p/2}(\Omega)}\leq C|D|^{2/3}\log(1+|D|)\,\norm{X_{0}}_{L^p(\Omega)}^2$$
and~\eqref{eq:estim_interaction_dipoles} is proved.
This concludes the proof of~\eqref{eq:simple_upper_bd}.

The bounds~\eqref{eq:upper-bd-average} follows from the Lieb-Thirring inequality and the stability of matter. By~\eqref{eq:stability-matter} with a $1/2$ in front of the kinetic energy instead of a $1$, we see that the total energy is bounded from below by
$$\cF_{T,\mu}(\omega,D)\geq \frac12\,\tr_{\mathscr{F}}\left(\sum_i(-\Delta)_i\Gamma(\omega)\right)-C|D|$$
almost surely. Hence our upper bound on $\cF_{T,\mu}(\omega,D)$ yields 
$$\norm{\tr_{\mathscr{F}}\left(\sum_i(-\Delta)_i\Gamma(\omega)\right)}_{L^{p/2}(\Omega)}\leq C|D|$$
for a regular domain $D$.
By the Lieb-Thirring inequality we have
$$\tr_{\mathscr{F}}\left(\sum_i(-\Delta)_i\Gamma(\omega)\right)\geq C\int_{D}\rho_{\Gamma(\omega)}^{5/3}\geq C|D|^{-2/3}\left(\int_D\rho_{\Gamma(\omega)}\right)^{5/3}$$
almost surely, which gives~\eqref{eq:upper-bd-average}. This ends the proof of Lemma~\ref{lem:upper_bound}.
\end{proof}

To simplify some estimates from below that we will derive later, we now introduce an auxiliary free energy $\underline{\cF}_{T,\mu}(\omega,D)$ obtained by minimizing over the charges $z$ of the nuclei $(R,z)\in\cK(\omega)\cap D$, which are at a distance $\leq2\epsilon$ from the boundary of $D$. This means we replace the charge $z(\omega)$ of each of these nuclei by $z'$ and we minimize over these $z'$s, under the constraints that $0\leq z'\leq z(\omega)$. This trick simplifies some lower bounds and it was also used in~\cite{HaiLewSol_2-09}. The idea is to show the existence of the thermodynamic limit for $\underline{\cF}_{T,\mu}$ and, only in the end, to prove that this implies the result for the original function ${\cF}_{T,\mu}$. This is done by using Lemma~\ref{lem:A4} below, and the fact that
\begin{equation}
\underline{\cF}_{T,\mu}(\omega,D) \leq {\cF}_{T,\mu}(\omega,D)
\label{eq:relation-opt-charges} 
\end{equation}
for all $D$ and almost all $\omega$. Note that the random variable $\underline{\cF}_{T,\mu}(\omega,D)$ satisfies the same stability of matter inequality~\eqref{eq:stability-matter} as ${\cF}_{T,\mu}(\omega,D)$, and by~\eqref{eq:relation-opt-charges} it satisfies the same upper bound~\eqref{eq:simple_upper_bd} as ${\cF}_{T,\mu}(\omega,D)$. It is also a stationary function in the sense of~\eqref{eq:stat-E-N}. The following is the equivalent of~\cite[Prop. 4]{HaiLewSol_2-09} in the random case.

\begin{lemma}[Control in average of charge variations at the boundary]\label{lem:A4}
Let $D'\in\cR_{a',\epsilon'}$ and $D\in\cR_{a,\epsilon}$ be two regular domains such that $D'\subset D$ and ${\rm d}(\partial D,\partial D')\geq C$. Then we have
\begin{equation}
\bE\big({\cF}_{T,\mu}(\cdot,D)\big)\leq \bE\big(\underline\cF_{T,\mu}(\cdot,D')\big)+ C|D\setminus D'|+C|D|^{13/15}
\label{eq:A4}
\end{equation}
where $C>0$ depend on $a$, $a'$, $\epsilon$, $\epsilon'$, $W$, $\mu$ and $T$, but not on $D$ and $D'$. 
\end{lemma}

The power $13/15$ is not optimal. In the end of the proof we indicate how to improve it.

\begin{proof}
Like for the proof of Lemma~\ref{lem:upper_bound} above which was based
on~\cite[Prop. 2]{HaiLewSol_2-09}, we now follow the proof
of~\cite[Prop. 4]{HaiLewSol_2-09}, but we keep track of the smallest
distance between the nuclei and we use a slightly different argument in
the end. To simplify our reasoning we assume that $\epsilon=\epsilon'$
and that $T=\mu=0$. We denote $\underline{\cF}:=\underline{\cF}_{0,0}$. The present proof easily carries over to the general case.

\medskip

\noindent\textit{$(i)$ The trial state.}
For any fixed $\omega\in\Omega$ we pick the exact trial state $\Gamma$ for the variational problem $\underline\cF(\cdot,D')$. We then use this trial state to get the upper bound~\eqref{eq:A4}. Recall that for the variational problem $\underline\cF(\cdot,D')$ the charges close to the boundary of $D'$ are optimized. In our system the nuclei do not necessarily have these optimal charges.
In $D\setminus D'$ we have several nuclei which we have to screen. Like in the proof of Lemma~\ref{lem:upper_bound}, we do this by placing electrons in small balls of radius $\delta'_{(R,z)}/4$ as close as possible to each nucleus. We put the electron on top of the nucleus if the nucleus is not too close to the boundary of $D\setminus D'$ and we place it closeby otherwise, thanks to the cone property. When the ball sits on top of the nucleus we call this a ``perfectly screened nucleus'' whereas we call the other ones ``dipoles''. For later purposes we have to make sure that only the nuclei which are very close to the boundary are not completely screened. So we choose
$$\delta'_{R,z}=\min\big(\delta_{R,z},\epsilon/20\big).$$
The factor $1/20$ has no real significance but it is here to ensure that in any cone of size $\epsilon=\epsilon'$ which is completely enclosed in $D\setminus D'$, there is always a ball of radius $\epsilon/5$ in which there cannot be any dipole.

Lastly, we have to cope with the fact that the charges in $D'$ close to the boundary of $D'$ do not have their optimal charge $z_\text{opt}$, but rather the normal charge $z=z_\text{opt}+\delta z$. This additional positive charge $\delta z$ might create important electrostatic errors and we also have to screen it by adding electrons outside of $D'$. In spirit we follow the technique of~\cite{HaiLewSol_2-09}. To any unit cell $W_j$ which is at a distance $\leq\epsilon$ to the boundary of $D'$, we associate a little cone of size $\epsilon$ in $D\setminus D'$, at a distance $\leq\epsilon$ to $W_j$. This cone only depends on $D$ and $D'$, it does not depend on the random variable $\omega$. In this cone we know that there is a ball of radius $\epsilon/5$ at a distance $\geq \epsilon/5$ to the boundary of the cone, hence also at a distance $\geq \epsilon/5$ to the boundary of $D'$. We put the screening electrons in a small ball $B_j$ of fixed radius $\sim\epsilon$ in the cone, at a
  distance $\geq\epsilon/5$ to the boundary of the cone. Their total charge must be equal to
$$\delta Z_j:=\sum_{(R,z)\in \cK\cap W_j}\delta z.$$
Note that each cone can intersect a finite (bounded) number of the other cones. This is because the cells $W_j$ are at distance $\leq 2\epsilon$ to their corresponding cone.
So in the cone we can always reduce the size of the balls in which we put the electrons, to make them all fit without any overlap. On the contrary to~\cite{HaiLewSol_2-09} where there was a smallest distance between all the nuclei, in our situation the additional electrons cannot always be chosen at a finite distance from all the other charges. However our construction guarantees that they can only overlap with perfectly screened nuclei, never with dipoles.

If we summarize the situation, we have in our system (see Figure~\ref{fig:trial})
\begin{itemize}
\item \emph{electrons in $D'$}, chosen to minimize the energy $\underline\cF(\omega,D')$, with the optimal charges close to the boundary;
\item \emph{nuclei in $D'$}. They have a charge which might be larger than the optimal one when they are close to the boundary of $D'$;
\item \emph{classical dipoles outside of $D'$}, at a finite distance $\leq \epsilon/10$ to the boundary of $D\setminus D'$;
\item \emph{electrons in balls of a fixed radius}, with a charge $\delta Z_j$ used to compensate the charges of some of the nuclei in $W_j\subset D'$. They are at a distance $\leq2\epsilon$ but $\geq \epsilon/5$ to the boundary of $D'$. They can never overlap with the dipoles;
\item \emph{perfectly screened nuclei} living in $D\setminus D'$, at a distance at least $\epsilon/10$ to the boundaries of $D$ and $D'$. They do not interact with anybody, except possibly with the radial electron which we might have added in order to compensate some charges in $D'$.
\end{itemize}

\begin{figure}[h]
\centering
\includegraphics{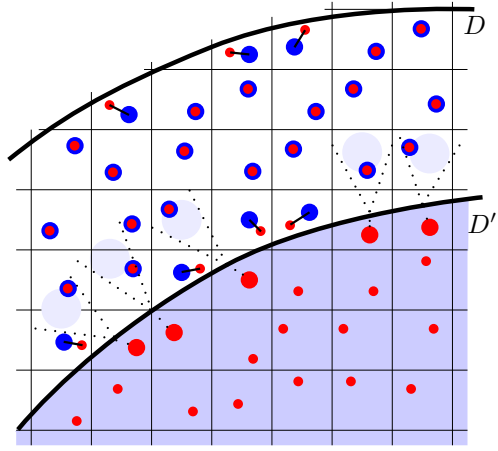}
\caption{The trial state used in the proof of Lemma~\ref{lem:A4}. The picture is here for a cubic lattice with always exactly one nucleus per unit cell.\label{fig:trial}}
\end{figure}

\medskip

\noindent\textit{$(ii)$ Estimates on the energy of the trial state.} 
We call $\cJ$ the set of all the indices such
that $W_j\cap (D\setminus D')\neq\emptyset$, $\partial\cJ$ the set of
indices such that $W_j$ intersects the boundary of $D\setminus D'$, and
$\partial\cJ_0$ the set of indices such that $W_j$ contains a nucleus whose charge has been optimized. For any $j_0\in \partial\cJ_0$, there is a ball $B_{j_0}$ outside of $D'$, containing the additional electrons with total charge $\delta Z_j$.

Now we estimate all the terms. First, we have to pay for the kinetic
energy to put the additional electrons in $D\setminus D'$. Recall that
we have two kinds of electrons, those which are squeezed in small balls
of radii $\delta'_{R,z}$ and those which are used to compensate the
charges of the nuclei close to the boundary of $D'$. The latter live in
a ball of a fixed radius. The total kinetic energy of all these
electrons, which we denote by $K^e_{D\setminus D'}$, is bounded from
above by
\begin{align*}
K_{D\setminus D'}^e&\leq C\sum_{j\in\cJ}\sum_{(R,z)\in\cK\cap W_j}\frac{z^{5/3}}{(\delta'_{R,z})^2}+C\sum_{j\in\cJ_0}\frac{(\delta Z_j)^{5/3}}{\epsilon^2}\\
&\leq C\sum_{j\in\cJ}X'_2(\tau_j\cdot)+C\sum_{j\in\cJ_0}X_0(\tau_j\cdot)^{5/3}.
\end{align*}
This inequality has been obtained by using the first eigenstates of
$-\Delta$ on each ball, computing the associated Hartree-Fock state, and
using it to define the electronic state in each ball.
Note the power $5/3$ which accounts for the fermionic nature of the electrons. Taking the average and using the regularity properties of $D$ and $D'$, we obtain the bound
$$\bE\big(K_{D\setminus D'}^e\big)\leq C\,|D\setminus D'|\; \bE\,X'_2+C\,|D'|^{2/3}\;\bE\,X_0^{5/3}\leq C\big(|D\setminus D'|+|D|^{2/3}\big).$$

Except for the kinetic energy, the nuclei which are completely screened do not participate much in our system, by Newton's theorem. There is only the possibility that they overlap with some electron in a ball $B_j$ with $j\in\cJ_0$. Let us denote by $\rho_j$ the corresponding electronic density in the ball $B_j$ and by $V_R$ the Coulomb potential induced by the nucleus $(R,z)$ together with its electron of size $\delta'_{(R,z)}/4$. We estimate the interaction between them using that each $V_R$ vanishes outside of the ball of radius $\delta'_{R,z}/4$:
\begin{align*}
\int\rho_j\sum_{(R,z)\in\cK\cap W_k}V_R&\leq
C\int_{B_j}\rho_j^{5/3}+C\int_{B_j}\left(\sum_{(R,z)\in\cK\cap B_j\cap W_k}V_R\right)^{5/2}\\
&\leq C\,X_0(\tau_j\omega)^{5/3}+C\sum_{(R,z)\in\cK\cap B_j\cap W_k}\int_{B(\delta'_{R,z})/4}\frac{z^{5/2}}{|x|^{5/2}}\\
&\leq C\,\left(X_0(\tau_j\omega)^{5/3}+\sum_{W_k\cap B_j\neq\emptyset} X_0(\tau_k\omega)\right).
\end{align*}
Here we have used the Lieb-Thirring inequality to control
$\int_{B_j}\rho_j^{5/3}$ by the kinetic energy, which in turn is bounded
by $X_0(\tau_j\omega)^{5/3}.$
We get a similar term for every $j\in\cJ_0$. Summing over such $j$'s, we deduce that the average of this error term is bounded above by a constant times
$|D|^{2/3}(\bE\,X_0^{5/3}+\bE\,X_0)$.

Our conclusion is that, up to an error of the form
$|D\setminus D'|+|D|^{2/3}$, we get the energy of a system in which we
only have the quantum electrons in $D'$, interacting with classical
particles. These are the nuclei in $D'$ (with charges which might be
higher than the optimal ones), as well as classical charges outside of
$D'$. The latter are the dipoles of charges $z$ and $-z$ at a distance
$\leq\epsilon/10$ to the boundaries of $D$ and $D'$, plus the additional
electrons of charges $\delta Z_j$, used to compensate some charges in
$D'$. We can write 
$$
\underline{\cF}(\cdot, D) \leq \cG +
\underline{\cF}(\cdot,D') + C|D\setminus D'| + C |D|^{2/3},
$$
with
\begin{equation}
\cG = \underline{\cF}(\cdot,D')-\int\rho_\Gamma (W_\text{in}+W_\text{out}) + \sum_{(R,z)\in\cK\cap D'}z\big(W_\text{in}(R)+W_\text{out}(R)\big)+I_{\text{class}} ,
\label{eq:form_energy_classical_particles}
\end{equation}
where $W_\text{out}$ is the Coulomb potential induced by all the classical particles sitting outside of $D'$, and 
$$W_\text{in}(x)=\sum_{j\in\cJ_0}\sum_{(R,z)\in\cK\cap W_j}\frac{\delta z}{|R-x|}$$
is the potential corresponding to the excess charges of the nuclei in $D'$. Finally, $I_{\text{class}}$ is the Coulomb interaction between all these classical charges.

In order to estimate $\cG$ in~\eqref{eq:form_energy_classical_particles}, we use the method of~\cite{HaiLewSol_2-09}. We write
\begin{align*}
\cG&=\underline{\cF}(\cdot,D')+\eta\bigg[-\underline{\cF}(\cdot,D')-\frac1\eta\int\rho_\Gamma (W_\text{in}+W_\text{out})\\
&\qquad + \frac1\eta\sum_{(R,z)\in\cK\cap D'}z\big(W_\text{in}(R)+W_\text{out}(R)\big)-\frac{I_{\text{class}}}{\eta^2}\bigg]\\
&\qquad +\eta\, \underline{\cF}(\cdot,D') +\left(1+\frac{1}{\eta}\right)I_{\text{class}}.
\end{align*}
The interpretation of the term in square bracket is that we have multiplied the charge of all the particles outside of $D'$ by a factor $-1/\eta$. Similarly we have changed the charges of the particles in $D'$ to $z_\text{opt}-\delta z/\eta$, instead of $z=z_\text{opt}+\delta z$. Now we claim that there is a stability of matter estimate in the form
\begin{multline}
\bE\,\bigg[\underline{\cF}(\cdot,D')+\frac1\eta\int\rho_\Gamma (W_\text{in}+W_\text{out})\\
 - \frac1\eta\sum_{(R,z)\in\cK\cap D'}z\big(W_\text{in}(R)+W_\text{out}(R)\big)+\frac{I_{\text{class}}}{\eta^2}\bigg]\geq -C|D|-C\frac{|D|^{2/3}}{\eta^{5/2}}.
\label{eq:stability_dipoles}
\end{multline}
We first explain how to use this estimate, before turning to its proof in Step $(iii)$. Inserting this and using that $\underline{\cF}(\cdot,D')\leq C|D|$ by Lemma~\ref{lem:upper_bound}, we get an estimate on the average of $\cG$:
$$\bE(\cG)\leq \bE\big(\underline{\cF}(\cdot,D')\big)+C\eta|D|+C\frac{|D|^{2/3}}{\eta^{3/2}}+C\frac{\bE\,\big(I_{\text{class}}\big)}{\eta}.$$
Using that the dipole-dipole interaction decays like $R^{-3}$ at infinity, the classical interaction term can be estimated following the proof of Lemma~\ref{lem:upper_bound} (see Eq.~\eqref{eq:estim_interaction_dipoles}):
$$\bE\;\big(I_{\text{class}}\big)\leq C|D|^{2/3}\log(|D|).$$
The final estimate on~\eqref{eq:form_energy_classical_particles} is
$$\bE(\cG)\leq \bE\big(\underline{\cF}(\cdot,D')\big)+C\eta|D|+C\frac{|D|^{2/3}}{\eta^{3/2}}+C\frac{|D|^{2/3}\log(|D|)}{\eta}.$$
Now if we optimize in $\eta$, finding that $\eta \propto |D|^{-2/15}$, and put back the other error terms, we arrive at our final estimate
\begin{equation}
\bE\big(\underline{\cF}(\cdot,D)\big)\leq \bE\big(\underline{\cF}(\cdot,D')\big)+C|D\setminus D'|+C|D|^{{13}/{15}}.
\label{eq:final-A4} 
\end{equation}

\medskip

\noindent\textit{$(iii)$ Proof of the stability of matter estimate~\eqref{eq:stability_dipoles}.}
The estimate~\eqref{eq:stability_dipoles} is the equivalent of~\cite[Lemma 10]{HaiLewSol_2-09} but, unfortunately, it does not follow from this result directly, because there it was again assumed that there is a smaller distance between all the nuclei. To cope with this issue we slightly change the argument of~\cite{HaiLewSol_2-09}.

The first step is the same as in~\cite{HaiLewSol_2-09} and it consists in replacing the Coulomb potential $1/|x|$ between all the particles by the Yukawa potential $e^{-|x|}/|x|$. First we know that the Fourier transform of the difference is positive:
$$\frac{1}{|k|^{2}}-\frac{1}{1+|k|^2}\geq0$$
and second we know that $|x|^{-1}-|x|^{-1}e^{-|x|}\to1$ when $|x|\to0$. All this implies that for any $y_i\in\R^3$ and any charges $q_i\in\R$,
$$\sum_{i\neq j}q_iq_j\left(\frac{1}{|y_i-y_j|}-\frac{e^{-|y_i-y_j|}}{|y_i-y_j|}\right)\geq -\sum_i q_i^2.$$
Thus, when we replace the Coulomb potential by the Yukawa potential, we get (in average) an error term of the form
$$-C\left(|D|+\frac{|D|^{2/3}}{\eta^2}\right).$$
The first term is an estimate on the average number of electrons as well as the average number of nuclei in $D'$. The second term is an estimate on the average number of classical particles which are close to the boundaries of $D$ and $D'$, and whose charge has been multiplied by $-1/\eta$.

Now that we have replaced the Coulomb interaction by Yukawa, the second step of the proof consists in dropping all the classical negative charges. In a lower bound we only pay for the interaction with the positive charges. Consider for instance the interaction between the negative classical particles $1/\eta$ inside or outside $D'$ and the nuclei in $D'$ which have the normal charge $z$ or the optimized charge $z_\text{opt}$. Because the Yukawa potential decays very fast, this interaction is easily controlled. It can be estimated similarly as in the proof of Lemma~\ref{lem:upper_bound} by
$$\leq \eta^{-1}\sum_{j\in(\partial\cJ)}\sum_{(R,z)\in W_j}\frac{1}{\delta'_{R,z}}+\eta^{-1}\sum_{j\in(\partial\cJ)}\sum_{k\in\cJ}X_0(\tau_j\omega)X_0(\tau_k\omega)e^{-|j-k|}$$
almost surely. The first term accounts for nuclei which are in neighboring cells, whereas the second one deals with nuclei which are in non-adjacent cells. Hence the average of this term is bounded above by
$\eta^{-1}\big(\bE(X_1)+\bE(X_0^2)\big)\#\partial\cJ\leq C\eta^{-1}|D|^{2/3}$. The argument is the same for the interaction between positive and negative charges $\eta^{-1}$ except that a crude bound gives $C\eta^{-2}|D|^{2/3}$. All in all, we see that when we throw away the negative charges, we make an error which is bounded from below by
$-C{|D|^{2/3}}{\eta^{-2}}$.

Now we have reduced ourselves to a system of electrons interacting with nuclei through the Yukawa potential, up to a total error of the form $-C(|D|+|D|^{2/3}\eta^{-2})$. The nuclei in $D'$ have a normal charge but the ones outside of $D'$ have the charge $z/\eta$. At this step we use the stability of matter with Yukawa, as was proved by Conlon, Lieb and Yau in~\cite{ConLieYau-88} through the Thomas-Fermi Yukawa energy. First we use the Lieb-Thirring inequality and the Lieb-Oxford-type bound 
\cite[Eq. (A.17)]{ConLieYau-88}, and estimate the quantum energy from below by Thomas-Fermi theory. Then we use the lower bound on the form $-CN-C\sum_i z_i^{5/2}$ which is proved in~\cite[Eq. (A.15)]{ConLieYau-88}. Hence, in average we get a lower bound of the form $-C|D|-C|D|^{2/3}\eta^{-5/2}$. Recall that the other error terms are not worse than $|D|^{2/3}\eta^{-2}$, which is itself smaller than $|D|^{2/3}\eta^{-5/2}$ for $\eta\ll1$.

Now that we have proved~\eqref{eq:stability_dipoles}, this concludes the proof of Lemma~\ref{lem:A4}.
\end{proof}

\begin{remark}
Using a Yukawa potential with mass $\mu$ and optimizing with respect to this mass in the end, it is possible to improve the error term $|D|^{13/15}$.
\end{remark}

\subsection*{Step 3. Lower bounds}
We already have one important lower bound on $\cF_{T,\mu}$ and $\underline\cF_{T,\mu}$, the one~\eqref{eq:stability-matter} corresponding to the stability of matter. This lower bound is true independently of $\omega$ and of the shape of the domain $D$, which need not be regular. 
The existence of the thermodynamic limit (for simplices at least) follows from a much more precise lower bound which is stated in the following lemma.

\begin{lemma}[Graf-Schenker type inequality]\label{lem:Graf-Schenker}
Let $\triangle\subset\R^3$ be a fixed simplex (a tetrahedron). 
Then we have the following lower bound
\begin{multline}
\underline{\cF}_{T,\mu}(\omega,D)\geq \left(1-\frac{C}{\ell}\right)\int_G\frac{\underline\cF_{T,\mu}(\omega,D\cap g\ell\triangle)}{|\ell\triangle|}\,dg\\
-\frac{C}\ell \Big(\#\big\{(R,z)\in\cK(\omega)\cap D\big\} + |D|\Big)
\label{eq:Graf-Schenker}
\end{multline}
for every domain $D$, every $\ell\geq1$, and with a universal constant $C$ which only depends on the chosen simplex $\triangle$.
\end{lemma}

We recall that $G=\R^3\rtimes SO(3)$ is the group of translations and rotations acting on $\R^3$, endowed with its Haar measure. Since we have by convention $\underline{\cF}_{T,\mu}(\omega,\emptyset)=0$, the integral over $G$ in the right side of~\eqref{eq:Graf-Schenker} can be restricted to a compact set and it is therefore convergent. The proof of~\eqref{eq:Graf-Schenker} is based on an important inequality of Graf and Schenker~\cite{GraSch-95} dealing with the Coulomb interaction of classical charges. This inequality was itself inspired of previous work by Conlon, Lieb and Yau~\cite{ConLieYau-88,ConLieYau-89} and it is recalled in~\cite[Sec. 1.1.2]{HaiLewSol_2-09}.

The proof of Lemma~\ref{lem:Graf-Schenker} is almost identical to that of~\cite{GraSch-95} and~\cite[p. 505--507]{HaiLewSol_2-09}.  We will not detail it again. The term $\#\{(R,z)\in\cK(\omega)\cap D\}$ comes from the control of a localization error term in the Graf-Schenker inequality~\cite{GraSch-95}. This estimate is conveniently done by using a version of stability of matter with Yukawa potentials which was derived in~\cite[Eq. (A5)]{ConLieYau-88}, based on a Thomas-Fermi-type theory. In~\cite{HaiLewSol_2-09} this error term was estimated by $\bar{Z}^{5/2}$ times the number of nuclei, the latter being itself bounded by $C|D|$ (when $D$ is regular). Here we do not have a smallest distance between the nuclei and we just keep the total number of nuclei in $D$. In this proof localizing the system to the small simplices $g\ell\triangle$ induces a small change of the nuclear charges close to the boundary of these simplices. This is why it is convenient to use the modified ener
 gy $\underline{
 \cF}_{T,\mu}$.

\begin{remark}\label{rmk:estim_error_GS}
The error term $\#\{(R,z)\in\cK(\omega)\cap D\}$ is not necessarily bounded uniformly with respect to $\omega$, but it has a limit almost-surely and in $L^p(\R^3)$, when we divide it by the `regularized' volume of $D$. Let us quickly explain this. For any $a,\epsilon>0$, we introduce similarly to~\cite[Eq. (9)]{HaiLewSol_1-09}
\begin{equation}
|D|_{a,\epsilon}:=\inf\big\{|D'|\ :\ D'\in\cR_{a,\epsilon},\ D'\supset D\big\}.
\label{eq:def_regularized_volume}
\end{equation}
One can verify that this volume is of the same order as the volume of the union of the sets $W+j$ with $j\in\ccL$ intersecting $D$.
Hence, for any fixed $a,\epsilon>0$, we have
\begin{equation*}
\#\big\{(R,z)\in\cK(\omega)\cap D\big\}\leq \sum_{\substack{j\in\ccL\ :\\ (W+j)\cap D\neq\emptyset}}\sum_{(R,z)\in \cK(\omega)\cap \overline{W+j}}1=\sum_{\substack{j\in\ccL\ :\\ (W+j)\cap D\neq\emptyset}}X_{0}(\tau_j\omega)
\end{equation*}
(recall~\eqref{eq:def_Delta_p}). We have $X_{0}\in L^{p}(\Omega)$, according to
assumption~\eqref{eq:assumption_nuclei} and Lemma~\ref{lem:borne-X0}.
By the ergodic theorem (Theorem~\ref{thm:Birkhoff} below), this term behaves like $|D|_{a,\epsilon}$ almost-surely and in $L^p(\Omega)$. In average we have the exact inequality
\begin{equation}
\bE\;\sum_{\substack{j\in\ccL\ :\\ (W+j)\cap D\neq\emptyset}}X_{0}(\tau_j\omega)\leq C\, \bE\big(X_{0}\big)\,|D|_{a,\epsilon}.
\label{eq:estim_average_error_GS}
\end{equation}
\end{remark}

The inequality~\eqref{eq:Graf-Schenker} (with a different error term)
was denoted by ${\rm (A5)}$
in~\cite{HaiLewSol_1-09,HaiLewSol_2-09}. There is a more precise but
much more complicated inequality ${\rm (A6)}$ which is explained at
length in these works. This inequality is also true in our case,
provided we take the expectation value.

\begin{lemma}[A more precise lower bound for simplices]
Let $D=gL\triangle$ be a dilated, rotated and translated simplex. Then the property ${\rm (A6)}$ of~\cite{HaiLewSol_1-09,HaiLewSol_2-09} is valid for the averaged free energy $\bE\big(\underline{\cF}_{T,\mu}(\cdot,D)\big)$.
\end{lemma}

The proof is again exactly the same as in~\cite{HaiLewSol_2-09}. Instead of using the inequality~\cite[Eq. (89)]{HaiLewSol_2-09} on the total number of electrons and on the kinetic energy in $D$, one uses~\eqref{eq:upper-bd-average}.

\medskip

From all the previous results we can now deduce that the thermodynamic limit exists for the deterministic function $\bE\big(\underline{\cF}_{T,\mu}\big)$, by simply applying the abstract main theorem of~\cite{HaiLewSol_1-09}.

\begin{corollary}[Thermodynamic limit in average]\label{cor:thermo-limit-average}
There exists a function ${f}(T,\mu)$ such that, for any sequence $(D_n)\subset \cR_{a,\epsilon}$ of regular domains with $a,\epsilon>0$, $|D_n|\to\ii$ and ${\rm diam}(D_n)|D_n|^{-1/3}\leq C$, we have
\begin{equation}
\lim_{n\to\ii}\frac{\bE\big(\cF_{T,\mu}(\,\cdot\,,D_n)\big)}{|D_n|}={f}(T,\mu).
\end{equation}
\end{corollary}

\medskip

\begin{proof}
The function $D\mapsto \bE\big(\cF_{T,\mu}(\,\cdot\,,D)\big)$ satisfies all the properties ${\rm (A1)}$--${\rm (A6)}$ of~\cite{HaiLewSol_1-09}. For ${\rm (A3)}$ it is even periodic,
$\bE\big(\cF_{T,\mu}(\,\cdot\,,D+k)\big)=\bE\big(\cF_{T,\mu}(\,\cdot\,,D)\big)$, $\forall k\in\ccL$,
by stationarity. We deduce the result from~\cite[Thm. 2]{HaiLewSol_1-09}. 
\end{proof}

Our task in the next step will be to upgrade the convergence of the average to a convergence in $L^q(\Omega)$. Since we already know from~\eqref{eq:simple_upper_bd} that $|D|^{-1}\underline\cF_{T,\mu}(\cdot,D)$ is bounded in $L^{p/2}(\Omega)$, strong convergence in $L^1(\Omega)$ will imply the result by interpolation.

\subsection*{Step 4. Thermodynamic limit in $L^1(\Omega)$}

This step is now more specific to the stochastic case. We will show that for a sequence $(D_n)$ of regular domains like in the statement of Theorem~\ref{thm:thermo-limit}, $|D_n|^{-1}\underline\cF_{T,\mu}(\cdot,D_n)$ converges strongly in $L^1(\Omega)$ to the same limit as its average, namely ${f}(T,\mu)$. 

We will make use of the following celebrated theorem.

\begin{theorem}[Ergodic theorem~\cite{Tempelman-72}]\label{thm:Birkhoff}
Let $X$ be a random variable in $L^p(\Omega)$ for some $1\leq p<\ii$.
Consider a sequence of sets $D_n\subset \R^3$ such that $|D_n|\to\ii$ and which is regular in the sense of Fisher~\cite{Fisher-64}, that is
\begin{equation}
\forall t\in[0,t_0],\quad \Big|\big\{x\in\R^3\ :\ {\rm d}(x,\partial D_n)\leq t\,|D_n|^{1/3}\big\}\Big|\leq |D_n|\,\eta(t)
\label{eq:Fisher}
\end{equation}
for some $t_0>0$ and some function $\eta$ with $\lim_{t\to0}\eta(t)=0$. Then we have 
$$\lim_{n\to\ii}\bE\left|\frac{1}{|D_n|}\sum_{k\in \ccL\cap D_n}X(\tau_k\omega)-\bE(X)\right|^p=0.$$
If moreover $D_n\subset B_{c|D_n|^{1/3}}$ for some $c>0$ and all $n$, then 
$$\frac{1}{|D_n|}\sum_{k\in \ccL\cap D_n}X(\tau_k\omega)\longrightarrow \bE(X)$$
almost-surely.
\end{theorem}

Theorem~\ref{thm:Birkhoff} can for instance be found in a paper of Tempel'man~\cite{Tempelman-72}. The Fisher regularity assumption is stronger than Tempel'man's condition $(E_1)$. The convergence in $L^p(\Omega)$ is Theorem 6.4$'$ of~\cite{Tempelman-72}, whereas the almost-sure convergence is Theorem~6.1 in~\cite{Tempelman-72}.

\medskip

Our first useful result is the following

\begin{lemma}[Limit of negative part]\label{lem:lower-bound-liminf}
Let $(D_n)\subset\cR_{a,\epsilon}$ be a sequence of regular domains like in the statement of Theorem~\ref{thm:thermo-limit}. Then we have
\begin{equation}
\lim_{n\to\ii}\bE\left[\frac{\underline\cF_{T,\mu}(\omega,D_n)}{|D_n|} - f(T,\mu)\right]_-=0.
\end{equation}
\end{lemma}

\begin{proof}
The proof mainly follows from the Graf-Schenker inequality~\eqref{eq:Graf-Schenker} which is our only precise bound valid almost-surely. We know that for all $\ell\geq1$
\begin{multline*}
\frac{\underline\cF_{T,\mu}(\omega,D_n)}{|D_n|}\geq \frac{1-C/\ell}{|D_n|}\int_G \frac{\underline\cF_{T,\mu}\big(\omega,D_n\cap(g\ell\triangle)\big)}{|\ell\triangle|}\,dg\\
-\frac{C}\ell\left(1+\frac1{|D_n|}\#\big\{(R,z)\in\cK(\omega)\cap D_n\big\}\right).
\end{multline*}
Note that we can write, as mentioned in Remark~\ref{rmk:estim_error_GS} and using the regularity of $D_n$, 
$$\frac1{|D_n|}\#\big\{(R,z)\in\cK(\omega)\cap D_n\big\}\leq \frac{C}{\#\cJ_n}\sum_{j\in\cJ_n} X_{0}(\tau_j\omega)$$
where $\cJ_n$ is the set of all points $j$ of the lattice $\ccL$ such
that $(W+j)\cap D_n\neq\emptyset$. By the ergodic theorem 
(Theorem~\ref{thm:Birkhoff}) and Lemma~\ref{lem:borne-X0}, we know that
\begin{equation}
\lim_{n\to\ii}\bE\left|\frac{1}{\#\cJ_n}\sum_{j\in\cJ_n} X_{0}(\tau_j\omega)-\bE\big(X_{0}\big)\right|=0.
\end{equation}
The limit also holds almost-surely provided that $D_n\subset B_{c|D_n|^{1/3}}$.
This means in particular that
\begin{equation*}
\frac{\underline\cF_{T,\mu}(\omega,D_n)}{|D_n|}\geq \frac{1-C/\ell}{|D_n|}\int_G \frac{\underline\cF_{T,\mu}\big(\omega,D_n\cap(g\ell\triangle)\big)}{|\ell\triangle|}\,dg-\frac{C}\ell(1+Y_n), 
\end{equation*}
where $Y_n\to0$ strongly in $L^1(\Omega)$.
Now we argue similarly as in~\cite[Proof of Thm.~1]{HaiLewSol_1-09}. We write the integral over translations and rotations as
\begin{multline}
\frac{1}{|D_n|}\int_G \frac{\underline\cF_{T,\mu}\big(\omega,D_n\cap(g\ell\triangle)\big)}{|\ell\triangle|}\,dg\\=\frac{1}{|D_n|}\sum_{j\in\ccL}\int_{W}d\tau\int_{SO(3)}dR\;\frac{\underline\cF_{T,\mu}\big(\omega,D_n\cap(R\ell\triangle+\tau+j)\big)}{|\ell\triangle|}. 
\label{eq:decompose-A5}
\end{multline}
We denote by $\cG_n$ the set of all $j\in\ccL$ such that $R\ell\triangle+\tau+j\subset D_n$ for all $\tau\in W$ and all $R\in SO(3)$. Similarly we denote by $\partial\cG_n$ the set of all the indices such that $(R\ell\triangle+\tau)+j\cap \partial D_n\neq\emptyset$ for at least one $R\in SO(3)$ or one $\tau\in W$. Using the regularity of $D_n$ we have $\#\partial\cG_n\leq C|D_n|^{2/3}$ and we obtain that
\begin{align*}
|D_n|&=\int_G \frac{|D_n\cap g\ell\triangle|}{|\ell\triangle|}\,dg\\
&=|W|(\#\cG_n)+\sum_{j\in\partial\cG_n}\int_W\,d\tau\int_{SO(3)}\,dR \frac{|D_n\cap(R\ell\triangle+\tau+j)|}{|\ell\triangle|}\\
&\leq |W|(\#\cG_n)+|W|(\#\partial\cG_n)\leq |W|(\#\cG_n)+C|D_n|^{2/3},
\end{align*}
see~\cite[Eq. (17)]{HaiLewSol_1-09}. Hence $\#\cG_n=|W|^{-1}|D_n|+o(|D_n|)$.
Coming back to~\eqref{eq:decompose-A5} and using the stability of matter
$$\underline\cF_{T,\mu}\big(\omega,D_n\cap(R\ell\triangle+\tau+j)\big)\geq -C|\ell\triangle|$$
for the indices $j\in\partial\cG_n$, we get
\begin{multline}
\frac{1}{|D_n|}\int_G \frac{\underline\cF_{T,\mu}\big(\omega,D_n\cap(g\ell\triangle)\big)}{|\ell\triangle|}\,dg\\
\geq \frac{1}{|W|}\int_{W}d\tau\int_{SO(3)}dR\; \left(\frac{1+o(1)}{\#\cG_n}\sum_{j\in\cG_n}\frac{\underline\cF_{T,\mu}\big(\tau_j\omega,R\ell\triangle+\tau)\big)}{|\ell\triangle|}\right)-C|D_n|^{-1/3}. 
\label{eq:decompose-A5-1}
\end{multline}
Therefore we have proved that
\begin{multline*}
\frac{\underline\cF_{T,\mu}(\omega,D_n)}{|D_n|}
\geq \frac{1-C/\ell}{|W|}\int_{W}d\tau\int_{SO(3)}dR\! \left(\frac{1+o(1)}{\#\cG_n}\sum_{j\in\cG_n}\frac{\underline\cF_{T,\mu}\big(\tau_j\omega,R\ell\triangle+\tau)\big)}{|\ell\triangle|}\right)\\
-C|D_n|^{-1/3}-\frac{C}\ell(1+Y_n).
\end{multline*}
We now subtract $f(T,\mu)$ on both sides, take the negative part and average over $\omega$. We obtain 
\begin{multline*}
\bE\left[\frac{\underline\cF_{T,\mu}(\cdot,D_n)}{|D_n|}-f(T,\mu)\right]_-
\leq C|D_n|^{-1/3}+\frac{C}\ell \big(1+\bE|Y_n|+f(T,\mu)\big)\\
+(1-C/\ell)\frac{1}{|W|}\int_{W}d\tau\int_{SO(3)}dR\; \bE\left|\frac{1+o(1)}{\#\cG_n}\sum_{j\in\cG_n}\frac{\underline\cF_{T,\mu}\big(\tau_j\,\cdot,R\ell\triangle+\tau)\big)}{|\ell\triangle|}-f(T,\mu)\right|.
\end{multline*}
By the ergodic theorem (Theorem~\ref{thm:Birkhoff}), we have 
$$\lim_{n\to\ii}\bE\left|\frac{1+o(1)}{\#\cG_n}\sum_{j\in\cG_n}\frac{\underline\cF_{T,\mu}\big(\tau_j\,\cdot,R\ell\triangle+\tau)\big)}{|\ell\triangle|}- \bE\frac{\underline\cF_{T,\mu}\big(\cdot,R\ell\triangle+\tau)\big)}{|\ell\triangle|}\right|=0.$$
By the dominated convergence theorem, we deduce that
\begin{multline*}
\lim_{n\to\ii}\frac{1}{|W|}\int_{W}d\tau\int_{SO(3)}dR\; \bE\left|\frac{1+o(1)}{\#\cG_n}\sum_{j\in\cG_n}\frac{\underline\cF_{T,\mu}\big(\tau_j\,\cdot,R\ell\triangle+\tau)\big)}{|\ell\triangle|}-f(T,\mu)\right|\\
=\frac{1}{|W|}\int_{W}d\tau\int_{SO(3)}dR\;\left|\bE\frac{\underline\cF_{T,\mu}\big(\cdot,R\ell\triangle+\tau)\big)}{|\ell\triangle|}-f(T,\mu)\right|.
\end{multline*}
Thus,
\begin{multline}
\limsup_{n\to\ii}\bE\left[\frac{\underline\cF_{T,\mu}(\omega,D_n)}{|D_n|}-f(T,\mu)\right]_-\\
\leq \frac{1}{|W|}\int_{W}d\tau\int_{SO(3)}dR\;\left|\bE\frac{\underline\cF_{T,\mu}\big(\cdot,R\ell\triangle+\tau)\big)}{|\ell\triangle|}-f(T,\mu)\right|+\frac{C}{\ell}\big(1+f(T,\mu)\big).
\label{eq:estim_almost_final}
\end{multline}
We have proved the existence of the thermodynamic limit for $\bE\,\underline\cF_{T,\mu}$ in Corollary~\ref{cor:thermo-limit-average}.
On the other hand we have a uniform bound
$$\left|\bE\left(\frac{\underline\cF_{T,\mu}\big(\cdot,R\ell\triangle+\tau)\big)}{|\ell\triangle|}\right)\right|\leq C$$
by Lemma~\ref{lem:upper_bound}. By the dominated convergence theorem, we conclude that the right side of~\eqref{eq:estim_almost_final} tends to 0 when $\ell\to\ii$, hence that
\begin{equation*}
\lim_{n\to\ii}\bE\left[\frac{\underline\cF_{T,\mu}(\omega,D_n)}{|D_n|}-f(T,\mu)\right]_-=0.
\end{equation*}
This concludes the proof of Lemma~\ref{lem:lower-bound-liminf}.
\end{proof}

\begin{remark}
\label{rmk:liminf1}
If we assume that $D_n\subset B_{c|D_n|^{1/3}}$, then we can actually show that
\begin{equation}
\liminf_{n\to\ii} \frac{\underline\cF_{T,\mu}(\omega,D_n)}{|D_n|}\geq f(T,\mu)
\label{eq:liminf_as}
\end{equation}
almost surely. Indeed, we have 
$$\frac{1}{\#\cG_n}\sum_{j\in\cG_n}\frac{\underline\cF_{T,\mu}\big(\tau_j\omega,R\ell\triangle+\tau)\big)}{|\ell\triangle|}\geq -C$$
by the stability of matter (Theorem~\ref{thm:stability}). Coming back
to~\eqref{eq:decompose-A5-1}, we can now use Fatou's Lemma and the
almost-sure ergodic theorem (Theorem~\ref{thm:Birkhoff}) to infer 
\begin{align*}
&\liminf_{n\to\ii}\frac{1}{|D_n|}\int_G \frac{\underline\cF_{T,\mu}\big(\omega,D_n\cap(g\ell\triangle)\big)}{|\ell\triangle|}\,dg\\
&\qquad\qquad\geq \frac{1}{|W|}\int_{W}d\tau\int_{SO(3)}dR\; \liminf_{n\to\ii}\left(\frac{1+o(1)}{\#\cG_n}\sum_{j\in\cG_n}\frac{\underline\cF_{T,\mu}\big(\tau_j\omega,R\ell\triangle+\tau)\big)}{|\ell\triangle|}\right)\\
&\qquad\qquad\geq \frac{1}{|W|}\int_{W}d\tau\int_{SO(3)}dR\;\; \bE\left(\frac{\underline\cF_{T,\mu}\big(\cdot,R\ell\triangle+\tau)\big)}{|\ell\triangle|}\right).
\end{align*}
The proof of~\eqref{eq:liminf_as} is then the same as before.
\end{remark}

Now we can deduce our desired result.

\begin{corollary}\label{cor:CV_L_1}
Let $(D_n)\subset\cR_{a,\epsilon}$ be a sequence of regular domains like in the statement of Theorem~\ref{thm:thermo-limit}. Then we have
\begin{equation}
\lim_{n\to\ii}\bE\left|\frac{\underline\cF_{T,\mu}(\omega,D_n)}{|D_n|} - {f}(T,\mu)\right|^q=0
\end{equation}
for $q=1$ if $p=2$ and all $1\leq q<p/2$ if $p>2$.
\end{corollary}

\medskip

\begin{proof}
We write
\begin{multline*}
\bE\left|\frac{\underline\cF_{T,\mu}(\cdot,D_n)}{|D_n|} - {f}(T,\mu)\right|\\
 =\bE\,\frac{\underline\cF_{T,\mu}(\cdot,D_n)}{|D_n|} - {f}(T,\mu)
+2\bE\left[\frac{\underline\cF_{T,\mu}(\cdot,D_n)}{|D_n|} - {f}(T,\mu)\right]_-.
\end{multline*}
The term on the right converges to 0 by Corollary~\ref{cor:thermo-limit-average} and Lemma~\ref{lem:lower-bound-liminf}.
Thus we get convergence in $L^1(\Omega)$.
In addition, we know that $\underline\cF_{T,\mu}(\cdot,D_n)|D_n|^{-1}$ is bounded in
$L^{p/2}(\Omega)$ by Lemma~\ref{lem:upper_bound}. Hence, by interpolation, we have
convergence in $L^q(\Omega)$ with $q=1$ if $p=2$ and $1\leq q<p/2$ if $p>2$.
\end{proof}

\begin{remark}
\label{rmk:liminf2} 
If $D_n\subset B_{c|D_n|^{1/3}}$, then we actually deduce that
\begin{equation}
\liminf_{n\to\ii} \frac{\underline\cF_{T,\mu}(\omega,D_n)}{|D_n|}= f(T,\mu)
\label{eq:liminf_as2}
\end{equation}
Indeed, we have by Fatou's Lemma (using that there is a uniform lower bound by the stability of matter)
\begin{equation*}
0\leq \bE\left(\liminf_{n\to\ii} \frac{\underline\cF_{T,\mu}(\omega,D_n)}{|D_n|}- f(T,\mu)\right)
\leq \liminf_{n\to\ii}\left(\bE\frac{\underline\cF_{T,\mu}(\omega,D_n)}{|D_n|}- f(T,\mu)\right)=0
\end{equation*}
and~\eqref{eq:liminf_as2} follows.
\end{remark}

\subsection*{Step 5. Back to the original free energy}
Our last step is to come back to the original energy $\cF_{T,\mu}(\omega,D)$ for which we do not optimize over the charges of the nuclei which are close to the boundary. First we have $\cF_{T,\mu}(\omega,D_n)\geq\underline\cF_{T,\mu}(\omega,D_n)$ by definition, hence obviously
$$\bE\left[\frac{\cF_{T,\mu}(\cdot,D_n)}{|D_n|} - f(T,\mu)\right]_-\leq \bE\left[\frac{\underline\cF_{T,\mu}(\cdot,D_n)}{|D_n|} - f(T,\mu)\right]_-\longrightarrow0.$$
by Lemma~\ref{lem:lower-bound-liminf}. Consider now a sequence $(D_n')\subset\cR_{a',\epsilon'}$ such that $D_n'\subset D_n$, $|D_n'|=|D_n|+o(|D_n|)$ and ${\rm d}(\partial D_n',\partial D_n)\geq 4\epsilon$. The sets $D_n'$ can for instance be constructed by considering a tiling of $\R^3$ and taking the union of all the cells in $D_n$ which are at the appropriate distance to the boundary of $D_n$. The regularity of such sets is proved in~\cite[Prop. 2]{HaiLewSol_1-09}. By Lemma~\ref{lem:upper_bound} we have 
$$\bE\left(\frac{\underline{\cF}_{T,\mu}(\cdot,D_n)}{|D_n|}\right)\leq \bE\left(\frac{{\cF}_{T,\mu}(\cdot,D_n)}{|D_n|}\right)\leq \bE\left(\frac{\underline{\cF}_{T,\mu}(\cdot,D'_n)}{|D'_n|}\right)+o(1)$$
which proves that 
$$\lim_{n\to\ii}\bE\left(\frac{{\cF}_{T,\mu}(\cdot,D_n)}{|D_n|}\right)=f(T,\mu).$$
The convergence in $L^q(\Omega)$ follows from the proof of Corollary~\ref{cor:CV_L_1}.
This concludes the proof of Theorem~\ref{thm:thermo-limit}.\qed

\section{Proof of Corollary~\ref{cor:neutrality}}\label{sec:proof_cor_neutrality}

Let us denote by $\Gamma_n(\omega)$ the optimal electronic state in Fock space for $\cF_{T,\mu}(\omega,D_n)$, and by $I_n(\omega)$ the total Coulomb interaction energy between the nuclei and the electrons in $D_n$. These quantities depend on $T$ and $\mu$ but we do not emphasize this in our notation. We have by the variational characterization~\eqref{eq:variational} of $\cF_{T,\mu}$ and by Lemma~\ref{lem:upper_bound}
\begin{equation*}
\bE\left( I_n+\tr\Big(\sum_i(-\Delta)_i-\mu\cN\Big)\Gamma_n+T\tr_\cF\Gamma_n\log\Gamma_n\right)=\bE\;\cF_{T,\mu}(\cdot,D_n)\leq C|D_n|.
\end{equation*}
We know from~\cite[Lemma 9]{HaiLewSol_2-09} that
\begin{align*}
\tr\left(\sum_i(-\Delta)_i-\mu\cN\right)\Gamma_n+T\tr_\cF\Gamma_n\log\Gamma_n
&\geq -T\log\left(\tr_\cF e^{-\big(\sum_i(-\Delta)_i-2\mu\cN\big)/2T}\right)\\
&\geq-C|D_n|.
\end{align*}
By the Lieb-Oxford inequality~\cite{LieOxf-80,LieSei-09} and the estimate~\eqref{eq:upper-bd-average} on $\bE\int\rho_{\Gamma_n}^{5/3}+\rho_{\Gamma_n}$, we have 
\begin{equation*}
\bE\,\tr\left(\sum_{k<\ell}\frac{1}{|x_k-x_\ell|}\right)\Gamma_n\geq \frac12 \bE\;{\rm D}(\rho_{\Gamma_n})-C\,\bE\;\int_{D_n}\rho_{\Gamma_n}^{4/3}\geq \frac12 \bE\;{\rm D}(\rho_{\Gamma_n})-C|D_n|
\end{equation*}
with
$${\rm D}(\rho_{\Gamma_n}):=\int_{D_n}\int_{D_n}\frac{\rho_{\Gamma_n}(x)\rho_{\Gamma_n}(y)}{|x-y|}\,dx\,dy$$
denoting the classical Coulomb energy. We deduce that
\begin{equation*}
\bE\left(\frac12 {\rm D}(\rho_{\Gamma_n}) -\int_{D_n}V_n\,\rho_{\Gamma_n}+U_n \right)\leq C|D_n|
\end{equation*}
where $V_n(\omega,x)$ is the electrostatic potential induced by the nuclei in $D_n$ and $U_n(\omega)$ is their Coulomb repulsion. At this step we replace each nucleus $(R,z)$ in $D_n\cap\cK(\omega)$ by a smooth, spherically symmetric, distribution $\nu_R$ leaving in a ball of radius $\min(1,\delta_{R,z}/3)$, the distance to the closest other nucleus. By Newton's theorem, $U_n$ does not change and the electrostatic potential $V_n(\omega,x)$ is modified only when $x$ is in these balls. The new potential $\tilde{V}_n(\omega,x)$ satisfies 
$$\Big|V_n(\omega,x)-\tilde{V}_n(\omega,x)\Big|\leq \overline{Z}\sum_{(R,z)\in D_n\cap\cK(\omega)}\frac{\1_{B(R,\min(1,\delta_{R,z}/3))}(x)}{|x|}.$$
Using that 
$$\left|\int_{D_n}\rho_{\Gamma_n}\big(\tilde{V}_n-V_n\big)\right|\leq \frac25\int_{D_n}\Big|V_n-\tilde{V}_n\Big|^{5/2}+\frac35\int_{D_n}\rho_{\Gamma_n}^{5/3}$$
and the estimate 
\begin{align*}
\int\Big|V_n-\tilde{V}_n\Big|^{5/2}\;&\leq \overline{Z}^{5/2}\sum_{(R,z)\in D_n\cap\cK(\omega)}\int_{B(R,\min(1,\delta_{R,z}/3))}\frac{dx}{|x|^{5/2}}\\
&\leq C \overline{Z}^{5/2}\sum_{(R,z)\in D_n\cap\cK(\omega)}1,
\end{align*}
we deduce from the Lieb-Thirring bound~\eqref{eq:upper-bd-average} and the fact that $\bE\,X_0<\ii$, that
\begin{equation*}
\frac12\bE\; {\rm D}\left(\rho_{\Gamma_n}-\sum_{(R,z)\in D_n\cap\cK(\cdot)}z\nu_R\right)\leq C|D_n|+ \frac12\bE\left(\sum_{(R,z)\in D_n\cap\cK} z^2\,{\rm D}\left(\nu_R\right)\right).
\end{equation*}
Since $D\left(\nu_R\right)=C/\delta_{R,z}$ and $\bE(X_1)<\ii$ by assumption, the last term is bounded by a constant times $|D_n|$ and we end up with the bound 
\begin{equation*}
\bE\; {\rm D}\left(\rho_{\Gamma_n}-\sum_{(R,z)\in D_n\cap\cK(\cdot)}z\nu_R\right)\leq C|D_n|.
\end{equation*}
At this step we use a capacity estimate, namely
$${\rm D}(f)\geq \left(\int_{D_n}f\right)^2\inf_{\substack{g\in H^{1}_0(D_n)\\ \int g=1}}{\rm D}(g)\geq \left(\int_{D_n}f\right)^2\inf_{\substack{g\in H^{1}_0(B_n)\\ \int g=1}}{\rm D}(g)=C\frac{\left(\int_{D_n}f\right)^2}{{\rm diam}(D_n)}$$
where $B_n$ is the smallest ball containing $D_n+B_1$. Using the regularity of $D_n$ as well as the fact that ${\rm diam}(D_n)\leq C|D_n|^{1/3}$, we get our final estimate
$$\bE\;\left|\frac1{|D_n|}\int_{D_n}\rho_{\Gamma_n} - \frac1{|D_n|}\sum_{(R,z)\in D_n\cap\cK(\cdot)}z\right|^2\leq \frac{C}{|D_n|^{2/3}}.$$
Since $X_0\in L^2(\Omega)$ and $z\in L^\ii(\Omega)$ by assumption, we
have by the ergodic theorem (Theorem~\ref{thm:Birkhoff})
$$\lim_{n\to\ii}\bE\;\left|\frac{1}{|D_n|}\sum_{(R,z)\in D_n\cap\cK(\cdot)}z-Z_{\rm av}\right|^2=0$$
and it follows that
$$\lim_{n\to\ii}\bE\;\left|\frac1{|D_n|}\int_{D_n}\rho_{\Gamma_n}-Z_{\rm av}\right|^2=0.$$
This is valid for any chosen temperature $T\geq0$ and chemical potential $\mu$.

To conclude the proof, we use the variational principle~\eqref{eq:variational} and get
$$\,\cF_{T,0}(\omega,D_n) - \mu\,{N_{T,\mu}(\omega,D_n)}\leq {\cF_{T,\mu}(\omega,D_n)}\leq {\cF_{T,0}(\omega,D_n)} - \mu\,{N_{T,0}(\omega,D_n)},$$
where we recall that $N_{T,\mu}(\omega,D_n)=\int{\rho_{\Gamma_n}}$ is the number of electrons at temperature $T\geq0$ and chemical potential $\mu$.
From the limit of $N_{T,\mu}(\cdot,,D_n)$ in $L^2(\Omega)$, hence in $L^1(\Omega)$, we finally conclude that
$$\lim_{n\to\ii}\bE\frac{\cF_{T,\mu}(\cdot,D_n)}{|D_n|}=\lim_{n\to\ii}\bE\frac{\cF_{T,0}(\cdot,D_n)}{|D_n|}- \mu\, Z_{\rm av}=f(T,0) - \mu\, Z_{\rm av}.$$
It is clear from the variational principle that $T\mapsto \bE\;\cF_{T,0}(\cdot,D_n)$ is concave, hence the limit $f(T,0)$ must also be concave. This concludes the proof of Corollary~\ref{cor:neutrality}.\qed

\bigskip

Let us end this paper by two final remarks.

\begin{remark}[A better convergence for simplices]
Like in~\cite{HaiLewSol_1-09,HaiLewSol_2-09}, it is possible to strengthen the result when $D_n$ is a sequence made of dilated simplices which are possibly rotated and translated. More precisely, we have 
$$\lim_{\ell\to\ii}\norm{\frac{\cF_{T,\mu}(\omega, g\ell\triangle)}{|\ell\triangle|}-f(T,\mu)}_{L^q(\Omega,L^\ii(G))}=0$$
for all $1\leq q<p/2$ if $p>2$ and $q=1$ if $p=2$.
\end{remark}


\begin{remark}[An abstract setting]
Our approach in this paper is general and it can be stated in an abstract fashion, like in~\cite{HaiLewSol_1-09}. Consider a random variable $\cF(\omega,D)$ defined on all bounded open subsets $D$ of $\R^3$. Assume that $\cF$ satisfies the following assumptions:

\medskip

\noindent {\rm (A1)}\qquad  $\cF(\omega,\emptyset)=0$ a.s.;

\medskip

\noindent {\rm (A2)}\qquad  $\cF(\omega,D)\geq -C$ for all $D$ and almost all $\omega$;

\medskip

\noindent {\rm (A3)}\qquad  $\cF(\omega,D+k)=\cF(\tau_k\omega,D)$ for all $D$, all $k\in\ccL$ and almost all $\omega$;

\medskip

\noindent {\rm (A5)}\qquad  $\displaystyle\cF(\omega,D)\geq \frac{1-\alpha(\ell)}{|\ell\triangle|}\int_{G}\cF(\omega,D\cap g\ell\triangle)\,dg - \big(|D|_{\alpha,\epsilon}+R_D(\omega)\big)\alpha(\ell)$ for all $D$ and almost all $\omega$, where $R_D$ is such that $\bE|R_D|\leq C|D|_{\alpha,\epsilon}$ for all $D$;

\medskip

\noindent {\rm (A4-6)}\quad  $D\mapsto\bE\,\cF(\cdot,D)$ satisfies the assumptions {\rm (A4)} and {\rm (A6)} of~\cite{HaiLewSol_1-09}.

\medskip

\noindent Under these conditions, the results of~\cite{HaiLewSol_1-09} as well as the method of the present article show that there exists a constant $f$ such that
$$\lim_{n\to\ii}\bE\left|\frac{\cF(\cdot,D_n)}{|D_n|}-f\right|=0$$
for any growing regular sequence $(D_n)$ like in Theorem~\ref{thm:thermo-limit}. Also
$$\liminf_{n\to\ii}\frac{\cF(\omega,D_n)}{|D_n|}=f\quad\text{a.s.}$$
when $D_n\subset B_{c|D_n|^{1/3}}$.
\end{remark}


\end{document}